\documentclass[english,aps, pra, twocolumn, superscriptaddress]{revtex4-1}
\usepackage[T1]{fontenc}
\usepackage[latin9]{inputenc}
\usepackage{color}
\usepackage{babel}
\usepackage{verbatim}
\usepackage{amsthm}
\usepackage{amsmath}
\usepackage{amssymb}
\usepackage{graphicx}
\usepackage{esint}
\usepackage[unicode=true,pdfusetitle,
 bookmarks=false,
 breaklinks=false,pdfborder={0 0 1},backref=false,colorlinks=true]
 {hyperref}

\makeatletter

\@ifundefined{textcolor}{}
{%
 \definecolor{BLACK}{gray}{0}
 \definecolor{WHITE}{gray}{1}
 \definecolor{RED}{rgb}{1,0,0}
 \definecolor{GREEN}{rgb}{0,1,0}
 \definecolor{BLUE}{rgb}{0,0,1}
 \definecolor{CYAN}{cmyk}{1,0,0,0}
 \definecolor{MAGENTA}{cmyk}{0,1,0,0}
 \definecolor{YELLOW}{cmyk}{0,0,1,0}
}
\theoremstyle{plain}
\newtheorem{thm}{\protect\theoremname}
  \theoremstyle{plain}
  \newtheorem{lem}[thm]{\protect\lemmaname}

\renewcommand{\fnum@figure}{FIG.~\thefigure}

\makeatother

  \providecommand{\lemmaname}{Lemma}
\providecommand{\theoremname}{Theorem}

\begin{document}

\title{Optimized tomography of continuous variable systems using excitation
counting}

\author{Chao~Shen}

\email{chao.shen@yale.edu}

\affiliation{Department of Applied Physics, Yale University, New Haven, CT 06511,
USA}

\author{Reinier~W.~Heeres }

\affiliation{Department of Applied Physics, Yale University, New Haven, CT 06511,
USA}

\author{Philip~Reinhold}

\affiliation{Department of Applied Physics, Yale University, New Haven, CT 06511,
USA}

\author{Luyao~Jiang}

\affiliation{Department of Physics, Yale University, New Haven, CT 06511, USA}

\author{Yi-Kai~Liu}

\affiliation{Joint Center for Quantum Information and Computer Science (QuICS),
University of Maryland, College Park, MD 20742, USA}

\affiliation{National Institute of Standards and Technology (NIST), Gaithersburg,
MD 20899, USA}

\author{Robert~J.~Schoelkopf}

\affiliation{Department of Applied Physics, Yale University, New Haven, CT 06511,
USA}

\affiliation{Department of Physics, Yale University, New Haven, CT 06511, USA}

\author{Liang~Jiang}

\affiliation{Department of Applied Physics, Yale University, New Haven, CT 06511,
USA}

\affiliation{Department of Physics, Yale University, New Haven, CT 06511, USA}
\begin{abstract}
\textcolor{black}{We propose a systematic procedure to optimize quantum
state tomography protocols for continuous variable systems based on
excitation counting preceded by a displacement operation. Compared
with conventional tomography based on Husimi or Wigner function measurement,
the excitation counting approach can significantly reduce the number
of measurement settings. We investigate both informational completeness
and robustness, and provide a bound of reconstruction error involving
the condition number of the sensing map. We also identify the measurement
settings that optimize this error bound, and demonstrate that the
improved reconstruction robustness can lead to an order-of-magnitude
reduction of estimation error with given resources. This optimization
procedure is general and can incorporate prior information of the
unknown state to further simplify the protocol. }
\end{abstract}
\maketitle
\global\long\def\bra#1{\left<#1\right|}
\global\long\def\ket#1{\left|#1\right>}
\global\long\def\overlap#1#2{\left<#1\mid#2\right>}
\global\long\def\norm#1{\left\Vert #1\right\Vert }
\global\long\def\matelem#1#2#3{\left<#1\right|#2\left|#3\right>}
\global\long\def\trace#1{\text{tr}#1}

\section{Introduction}

Quantum state tomography (QST) is a powerful procedure to completely
characterize quantum states, which can be extended to quantum process
tomography for general quantum operations. However, QST is often resource
consuming, involving preparation of a large number of identical unknown
states and measurement of a large set of independent observables.
For qubit systems, many techniques have been developed to reduce the
cost of full state tomography, such as compressed sensing \cite{CS_Gross_prl,CS_Flammia_njp,Kalev_npjquant2015},
permutationally invariant tomography \cite{permu_invar}, self-guided
or adaptive tomography \cite{self_guided_qubits,adaptive_tomo}, and
matrix product states tomography \cite{mps_tomo}. \textcolor{black}{In
contrast, for continuous variable (CV) systems that also play an important
role in quantum information, the standard techniques in use today
are decades old, namely homodyne measurement \cite{homodyne_wigner,info_completeness}
for optical photons and direct Wigner function measurement \cite{direct_wigner_Lutterbach,direct_wigner_Haroche,CV_cat}
for cavity QED. With the rapid development in CV quantum information
processing, ranging from arbitrary state preparation \cite{direct_wigner_Martinis}
to universal quantum control \cite{krastanov,Reineer_PRL_15} and
from engineered dissipation \cite{cat_encoding_njp,engin_dissip}
to quantum error correction \cite{cat_QEC_prl,Ofek_nature16}, a large
dimension of Hilbert space can be coherently controlled in experiments
\cite{CV_cat,Wang_science16}. However, homodyne measurement might
not be immediately applicable due to intrinsic nonlinearity preventing
applying a very large displacement in cavity QED, and Wigner function
measurement requires intensive data collection \cite{Wang_science16}.
Thus there is an urgent need for reliable and efficient tomography
for CV systems.}

\textcolor{black}{}

\textcolor{black}{There have been significant advances in excitation
counting over various physical platforms, including optical photons
\cite{opt_photon_num_Nam}, microwave photons \cite{book_Haroche,Haroche_1996,photon_num_haroche_nature,Martinis_qn_tomo_ring},
and phonons of trapped ions \cite{qn_tomo_Wineland,KihwanKim_Jarzynski2015,JHome}.
In particular, the capability of quantum non-demolition measurement
of microwave excitation number has been demonstrated with superconducting
circuits \cite{photon_num_cqed_nature}.} Tomography based on excitation
counting has also been theoretically proposed \cite{qn_tomo_welsch,Mancini_epl_1997}
and experimentally demonstrated with trapped ions and cavity or circuit
QED \cite{qn_tomo_Wineland,Martinis_qn_tomo_ring,Haroche_tomography_N}.
However, all these works only considered specific choices of measurement
settings (associated with certain displacement patterns), and mostly
restricted to the feasibility of tomography, without further investigating
the robustness against measurement noise to develop robust QST protocols
for CV systems.

\textcolor{black}{Motivated by these recent advances, we develop a
theoretical framework to investigate cost-effective QST protocols
for CV systems based on excitation counting. Conventional QST protocols
can be regarded as special cases collecting }\textit{\textcolor{black}{partial}}\textcolor{black}{{}
information of the excitation number distribution. For example, up
to a displacement, the Husimi $Q$ function can be regarded as the
probability of zero excitation, and the Wigner function can be obtained
from the difference between probabilities associated with even and
odd number of excitations. We expect more cost-effective QST by collecting
full population distributions upon various displacements using excitation
counting, which can be efficiently achieved in various CV systems
\cite{qn_tomo_Wineland,KihwanKim_Jarzynski2015,JHome,opt_photon_num_Nam,book_Haroche,Haroche_1996,photon_num_haroche_nature,photon_num_cqed_nature,Martinis_qn_tomo_ring}. }

The rest of the paper is organized as follows.\textcolor{black}{{} In
Sec. \ref{sec:IC_abstract}, we first provide a mathematical formulation
of QST based on displacements and excitation counting. We then consider
QST for a special class of quantum states in Sec. \ref{sec:QST-cats},
illustrating the advantage of excitation counting and introducing
the criterion of error robusteness in terms of the }\textit{\textcolor{black}{condition
number}}\textcolor{black}{{} (CN) of the sensing map in Sec. \ref{sec:Error-robustness}.
The main results on QST of a general unknown quantum state are presented
in Secs. \ref{sec:gen-states-IC} and \ref{sec:gen-state-robustness}.
In Sec. \ref{sec:noise-model}, the choice of optimization target
for different error models are analyzed.  We put our optimized scheme
to the test using simulated measurement records in Sec. \ref{sec:benchmarking}.
Section \ref{sec:generalizations} discusses possible generalizations
of the scheme. Finally, the conclusion is given in Sec. \ref{sec:Conclusion}. }

\section{Informational completeness \label{sec:IC_abstract}}

\textcolor{black}{Mathematically, QST solves the inversion problem
\[
A\cdot\vec{\rho}=\vec{b},
\]
where $\vec{\rho}$ is the unknown density matrix arranged as a vector,
$\vec{b}$ denotes all the measurement records, and $A$ is the sensing
matrix determined by the kind of measurements performed. }The set
of measurements should be \textit{informationally complete} (IC),
that is, the sensing matrix $A$ should be invertible %
\footnote{Unless there are additional constraints to $\rho$ so that other methods
like compressed sensing may apply%
}. For a non-square sensing matrix, the unknown state can be reconstructed
using least-squares fitting, 
\[
\vec{\rho}=\tilde{A}^{-1}\vec{b}=\left(A^{\dagger}A\right)^{-1}A^{\dagger}\vec{b}.
\]
\textcolor{black}{Due to experimental noise, the least-squares solution
may turn out non-physical, i.e., having negative eigenvalues. This
can be fixed by finding the physical density matrix $\sigma$ that
is closest to $\rho$, with the distance defined by some matrix norm,
e.g., the Frobenius norm. A justification of this procedure is provided
in Appendix \ref{sec:appdx-physical-rho}. }

\textcolor{black}{For CV systems, each measurement setting is associated
with a displacement operation $D\left(\beta\right)$. We may directly
count the excitation number after the displacement operation and obtain
the number distribution, which is called the }\textit{\textcolor{black}{generalized
Q function }}\textcolor{black}{($Q_{n}$ function)}\textit{\textcolor{black}{{}
}}\textcolor{black}{\cite{qn_tomo_welsch,Walmsley2009,Walmsley2012,photon_num_cqed_nature}}\textit{\textcolor{black}{{}
}}\textcolor{black}{
\begin{eqnarray*}
Q_{n}^{\beta}\left(\rho\right) & = & \mathrm{tr}\left[\ket n\bra nD(-\beta)\rho D^{\dagger}(-\beta)\right],
\end{eqnarray*}
where $n=0,\,1,\,2,\,\cdots,\, n_{c}$ with $n_{c}$ the maximal resolved
excitation number. }Reshaping $\rho$ into a column vector $\vec{\rho}$
we obtain the linear equation $\vec{Q^{\beta}}(\rho)=A^{\beta}\vec{\rho}$,
where $\vec{Q}^{\beta}(\rho)$ is a column vector with $(n_{c}+1)$
entries $Q_{n}^{\beta}(\rho)$ and the matrix $A^{\beta}$ has $(n_{c}+1)$
rows. Multiple measurement settings associated with a set of displacements
$\left\{ \beta_{1},\,\beta_{2},\,\cdots,\,\beta_{N_{\beta}}\right\} $
are used to constrain the inversion problem. The measurement record
$\vec{b}$ is then a column vector with $N_{\beta}\cdot(n_{c}+1)$
entries of $Q_{n}^{\beta_{j}}\left(\rho\right)$; the sensing matrix
$A$ can be obtained by stacking $A^{\beta_{i}}$ , with a total of
$N_{\beta}\cdot(n_{c}+1)$ rows. The basis under which $\rho$ is
written can be arbitrary, e.g. Fock basis $\ket{m_{1}}\bra{m_{2}}$
or coherent-state basis $\ket{\alpha_{i}}\bra{\alpha_{j}}$. 

\textcolor{black}{In comparison, the sensing matrix for standard QST
with the Husimi $Q$ function $Q_{n=0}^{\beta}\left(\rho\right)=\bra{\beta}\rho\ket{\beta}$
or Wigner function $W^{\beta}\left(\rho\right)=\sum_{n}\left(-1\right)^{n}Q_{n}^{\beta}\left(\rho\right)$
consists of only} $N_{\beta}$ rows {[}which are linear combinations
of $N_{\beta}\cdot(n_{c}+1)$ rows of the sensing matrix associated
with \textcolor{black}{the $Q_{n}$ function }%
\footnote{\textcolor{black}{Actually, in some experiments the Wigner function
was obtained from $Q_{n}^{\beta}$ \cite{direct_wigner_Martinis,photon_num_cqed_nature}.}%
}\textcolor{black}{{]}, which neglect a large portion of potentially
useful information. In the following, we consider QST for a class
of quantum states and show that the neglected information can be crucial.}

\section{\textcolor{black}{QST for cat states}}

\textcolor{black}{\label{sec:QST-cats}Cat states are quantum states
characterized by density matrix $\rho=\sum_{i,j=1}^{p}\rho_{ij}\ket{\alpha_{i}}\bra{\alpha_{j}}$,
where $\ket{\alpha_{i}}$ are well-separated coherent states \cite{QC_map_PRA_2013}.
The Schrödinger cat state $\ket{\alpha}+\ket{-\alpha}$ is a well-known
example. Standard QST of cat states with large unknown $\alpha$'s
is resource consuming and requires many measurement settings. In particular,
both the Husimi and Wigner function measurement schemes encounter
the challenge of unknown $\alpha$'s, and have to deploy many measurement
settings to scan various displacements, the majority of which is unfortunately
wasted because $Q^{\beta}(\rho)\approx W^{\beta}(\rho)\approx0$ for
most choices of $\beta$. In contrast, the $Q_{n}$ function measurement
always generates an excitation distribution, from which we can estimate
the distances $\left|\alpha_{i}-\beta\right|$ for different $\beta$.
Using the idea of trilateration, we can estimate all $\alpha$'s using
about }\textit{\textcolor{black}{three}}\textcolor{black}{{} measurement
settings. }Using the data $Q_{n}^{\beta}(\rho)$ for $\{\beta_{1},\,\beta_{2},\,\beta_{3}\}$,
we can estimate the density matrix $\tilde{\rho}$ using the iterative
maximum likelihood estimation (iMLE) technique \cite{iMLE} and calculate
the corresponding Husimi $Q$ function {[}see Fig. \ref{fig:iMLE-demon}
(b){]}. To increase confidence, one can additionally measure $Q_{n}^{\beta}(\rho)$
at one or two $\beta's$, preferrably at the current estimated $\alpha_{i}'s$
{[}see Fig. \ref{fig:iMLE-demon} (c), (d){]}. If the true state is
not a cat state, we would not see clearly separated population patches
in the phase space and need to treat it as a general state. 

\begin{figure*}
\centering{}\includegraphics{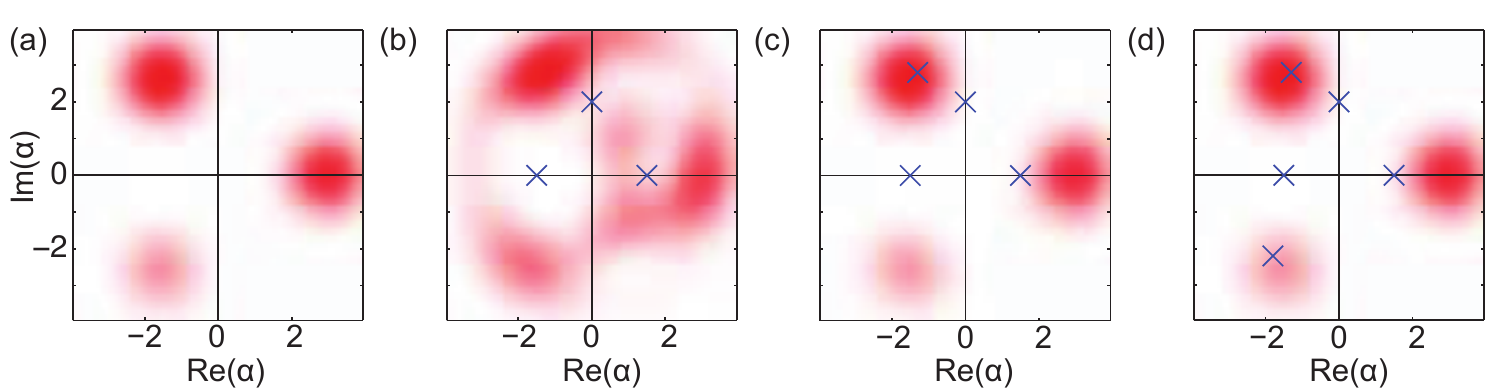}\protect\caption{Procedure of estimating the $\alpha_{i}$ via Husimi $Q$ function.
(a) shows the true $Q$ function of the state; (b) shows the estimated
$Q$ function via iMLE after measuring $Q_{n}^{\beta}(\rho)$ at three
$\beta's$ shown as the crosses; (c)/(d) are estimations after measuring
at four and five $\beta's$, respectively. Apparently the estimate
in (c) already converges to the true $Q$ function shown in (a). \label{fig:iMLE-demon}}
\end{figure*}

\textcolor{black}{Once the $\alpha's$ are known, the generalized
$Q$ function measurement only requires }\textit{\textcolor{black}{one
additional}}\textcolor{black}{{} measurement setting to fulfill the
IC requirement, independent of the number of coherent components.
It is note-worthy that examples where tomography requires only one
measurement setting are extremely rare. This observation can be justified
by the relation
\begin{eqnarray*}
Q_{n}^{\beta}(\rho) & = & \sum_{i,j=1}^{p}\rho_{ij}Q_{n}^{\beta}\left(\ket{\alpha_{i}}\bra{\alpha_{j}}\right)\\
 & = & \sum_{i,j=1}^{p}\rho_{ij}\mathrm{tr}\left[\ket n\bra nD(-\beta)\ket{\alpha_{i}}\bra{\alpha_{j}}D^{\dagger}(-\beta)\right]\\
 & = & \sum_{i,j=1}^{p}\rho_{ij}e^{i\theta(\beta,\alpha_{i},\alpha_{j})}e^{-\frac{1}{2}\left(\left|\alpha_{i}-\beta\right|-\left|\alpha_{j}-\beta\right|\right)^{2}}\\
 &  & \times\frac{1}{n!}\left[(\alpha_{i}-\beta)(\alpha_{j}-\beta)^{*}\right]^{n}e^{-\left|\alpha_{i}-\beta\right|\cdot\left|\alpha_{j}-\beta\right|}\\
 & = & \sum_{i,j=0}^{p}\tilde{\rho}_{ij}\frac{1}{n!}\left[d_{i}d_{j}e^{i\phi_{ij}}\right]^{n},
\end{eqnarray*}
where we defined 
\begin{eqnarray*}
d_{i} & \equiv & \left|\alpha_{i}-\beta\right|,\\
\phi_{ij} & \equiv & \arg(\alpha_{i}-\beta)-\arg(\alpha_{j}-\beta),\\
\theta(\beta,\alpha_{i},\alpha_{j}) & \equiv & -i(-\beta\alpha_{i}^{*}+\beta^{*}\alpha_{i}-\alpha_{j}\beta^{*}+\alpha_{j}^{*}\beta)/2,\\
\tilde{\rho}_{ij} & \equiv & e^{i\theta(\beta,\alpha_{i},\alpha_{j})}e^{-\frac{1}{2}(d_{i}-d_{j})^{2}}e^{-d_{i}d_{j}}\rho_{ij}.
\end{eqnarray*}
 Reshaping $\tilde{\rho}_{ij}$ as a column vector, we have
\[
\left(\begin{array}{cccc}
1 & \cdots & 1 & \cdots\\
\vdots & \ddots & \vdots\\
d_{1}^{2n} & \cdots & (d_{i}d_{j}e^{i\phi_{ij}})^{n}\\
\vdots &  & \vdots & \ddots
\end{array}\right)\left(\begin{array}{c}
\tilde{\rho}_{11}\\
\vdots\\
\tilde{\rho}_{ij}\\
\vdots
\end{array}\right)=\left(\begin{array}{c}
0!Q_{0}^{\beta}\\
\vdots\\
n!Q_{n}^{\beta}\\
\vdots
\end{array}\right).
\]
The matrix on the left-hand side is a Vandermonde matrix, having full
column rank (all column vectors are independent and $A^{\dagger}A$
is invertible) if and only if all $d_{i}d_{j}e^{i\phi_{ij}}$ are
distinct. }Under the following conditions, all the $d_{i}d_{j}e^{i\phi_{ij}}$
are distinct:\textcolor{black}{{} (i) $d_{i}\neq d_{j}$, other wise
the columns corresponding to $\tilde{\rho}_{ii}$ and $\tilde{\rho}_{jj}$
would be identical;; (ii) $\phi_{ij}\neq0,\,\pi$, otherwise the columns
$\tilde{\rho}_{ij}$ and $\tilde{\rho}_{ji}$ would be identical;
and (iii) $d_{i}d_{j}\neq d_{k}d_{l}$ or $\phi_{ij}\neq\phi_{kl}$
where all of $i,\, j,\, k,\, l$ are assumed to be distinct. These
requirements have clear geometric interpretations: (i) $\beta$ does
not lie on the perpendicular bisector of the line segment $\alpha_{i}\alpha_{j}$;
(ii) $\beta$, $\alpha_{i}$, $\alpha_{j}$ are not collinear; and
(iii) triangles formed by $(\beta,\,\alpha_{i},\,\alpha_{j})$ and
$(\beta,\,\alpha_{k},\,\alpha_{l})$ do not have the same area or
the angles subtended by the segments $\overline{\alpha_{i}\alpha_{j}}$
and $\overline{\alpha_{k}\alpha_{l}}$ from $\beta$ are different.
There is in fact one extra soft requirement, due to the factor $e^{-\frac{1}{2}(d_{i}-d_{j})^{2}}$
in $Q_{n}^{\beta}\left(\ket{\alpha_{i}}\bra{\alpha_{j}}\right)$.
When $d_{i}\ll d_{j}$ or $d_{i}\gg d_{j}$, $\rho_{ij}$ gets exponentially
suppressed and almost vanishes from the sensing equation, just like
the case with the conventional Husimi $Q$ function. So we add one
requirement (iv) $\beta$ does not lie far away from the bisector
of $\alpha_{i}\alpha_{j}$ in the sense that $e^{-\frac{1}{2}(d_{i}-d_{j})^{2}}$
is not too small. Requirement (iv) is closely related to the error
robustness which will be discussed later.} \textcolor{black}{The $Q_{n}$
function at one suitable $\beta$ contain sufficient information.}
More specifically, the diagonal terms in the density matrix $\rho_{ii}$
(the population of $\ket{\alpha_{i}}$) can be extracted from the
envelope of the distribution, while the off-diagonal terms $\rho_{i,j}$
can be obtained from the interference signals peaked at $\bar{n}=d_{i}d_{j}$
in the distribution. Therefore, sampling the excitation number distribution
can boost the information gain and thus reduce the measurement settings
significantly.

\section{Error robustness of reconstruction\label{sec:Error-robustness}}

So far, we have only considered the requirement for the IC, or possibility
of reconstruction. We do not yet know the accuracy of the reconstruction
when measurements are noisy. Next, we investigate robustness and estimate
the reconstruction error. Assume that the measurements $\vec{b}$
have noise $\delta\vec{b}$, leading to noise in the solution $\tilde{A}^{-1}\delta\vec{b}$.
To bound the noise in the solution, we consider the worst-case noise
magnification ratio
\[
\kappa(A)\equiv\frac{\left\Vert \tilde{A}^{-1}\delta\vec{b}\right\Vert /\left\Vert \tilde{A}^{-1}\vec{b}\right\Vert }{\left\Vert \delta\vec{b}\right\Vert /\left\Vert \vec{b}\right\Vert },
\]
which is called the CN of $A$ \cite{matrix_analysis}. The CN is
a property of the sensing map and does not depend on the specific
procedure that solves the linear equations. In principle the norm
can be chosen arbitrarily. We will use the  two-norm $\left\Vert \bullet\right\Vert _{2}$
of vectors, because in this case the CN is simply the ratio of the
largest and smallest singular values of $A$ \cite{matrix_analysis}.
Clearly $\kappa(A)\ge1$ and when $\kappa(A)=1$ the sensing map is
isometric (distance preserving). The CN has been introduced as a measure
of robustness of reconstruction schemes for qubit systems \cite{cond_num_Russian,opt_two_qbt_tomo_cond_num,opt_nuc_spin_cond_num}.
Using Uhlmann's definition 
\[
F(\rho,\,\sigma)=\text{Tr}\left[\sqrt{\sqrt{\rho}\sigma\sqrt{\rho}}\right],
\]
 the reconstruction fidelity can be bounded as (see Appendix \ref{sec:appdx-error-bound}
for a proof) 
\begin{equation}
F(\rho,\,\rho+\delta\rho)\ge1-\frac{1}{2}\kappa(A)\sqrt{r}\left\Vert \rho\right\Vert _{F}\left\Vert \delta\vec{b}\right\Vert _{2}/\left\Vert \vec{b}\right\Vert _{2},\label{eq:error_bound}
\end{equation}
where $r$ is the rank of $\delta\rho$ bounded by the system dimension,
and $\left\Vert \rho\right\Vert _{F}$ is the Frobenius norm of the
true density matrix which is fixed. Assuming for now that $\left\Vert \delta\vec{b}\right\Vert _{2}/\left\Vert \vec{b}\right\Vert _{2}$
is fixed (e.g., due to systematic bias), a robust QST should minimize
CN to have an optimal guarantee of the reconstruction fidelity. Note
that a lower CN reduces the sample complexity but not the computational
complexity of the inversion problem. 

\begin{figure}
\centering{}\includegraphics{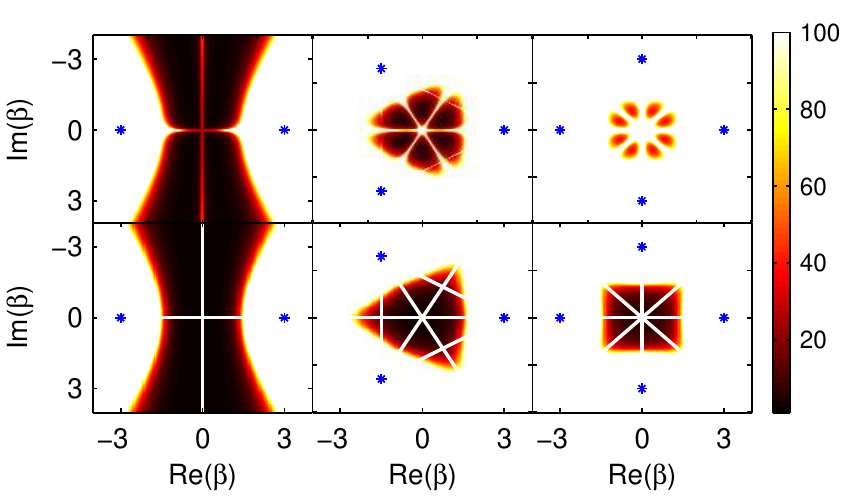}\protect\caption{(Color online)\textcolor{black}{{} Condition number of the sensing map
as a function of $\beta$ for cat states with number of components
$p=2,\,3,\,4$. Upper panels: numerical results for CN; Lower panels:
a simple estimate of the CN using the expression $\kappa(\beta)\sim\sum_{i,j}\exp\left[\left(d_{i}-d_{j}\right)^{2}/2\right]$
where $d_{i}\equiv\left|\alpha_{i}-\beta\right|$. We also included
the white lines on which the sensing map is strictly informationally
incomplete (see main text). Blue stars indicate the positions of the
coherent components $\protect\ket{\alpha_{i}}$. For visual clarity,
values beyond 100 are all mapped to white. The minimum CNs achievable
for the three cases are 1.74, 6.81, and 38.64 (numerical results),
respectively. Here the maximal resolved excitation number $n_{c}$
is taken sufficiently large. If $n_{c}$ decreases,} CN for large
$\left|\beta\right|$ gets worse. \label{fig:kappa_vs_beta}}
\end{figure}

\textcolor{black}{We now use CN to examine the robustness of QST for
cat states, for which CN is a function of one complex variable $\beta$.
Due to the factor $e^{-\frac{1}{2}(d_{i}-d_{j})^{2}}$ in $Q_{n}^{\beta}\left(\ket{\alpha_{i}}\bra{\alpha_{j}}\right)$,
when $d_{i}\ll d_{j}$ or $d_{i}\gg d_{j}$, $\rho_{ij}$ gets exponentially
suppressed, just like the case with the Husimi $Q$ function.} In
those regions, the factor $\exp\left[\left(d_{i}-d_{j}\right)^{2}/2\right]$
would magnify the noise during the reconstruction. Thus we estimate
\textcolor{black}{
\[
\kappa(\beta)\sim\sum_{i,j}\exp\left[\left(d_{i}-d_{j}\right)^{2}/2\right],
\]
}which agrees well with the numerical calculation of CN, as illustrated
in Fig.~\ref{fig:kappa_vs_beta}. Different from the requirement
for IC, CN depends on the number of coherent components $p$, the
values of $\alpha_{i}$, and the choice of $\beta$. For small $p$,
there exist low-CN regions of $\beta$ (dark regions in Fig.~\ref{fig:kappa_vs_beta}),
which imply that the protocol with only about four measurement settings
(about three for trilateration and one for coherences) can be robust. 

These low-CN regions are very similar to the regions with high Fisher
information in the worst case. For the state $\rho=\sum_{i,j=1}^{p}\rho_{ij}\ket{\alpha_{i}}\bra{\alpha_{j}}$
with known $\alpha_{i}$, the parameters to estimate are $\rho_{ij}$.
For convenience we arrange the $p^{2}$ numbers as a vector $\vec{\rho}$.
For a certain measurement position $\beta$, we can get a distribution
\[
f(n)\equiv Q_{n}^{\beta}(\vec{\rho}).
\]
According to the definition, the Fisher information matrix is 
\begin{eqnarray*}
\mathcal{I}(\vec{\rho}) & = & \mathbb{E}_{\vec{\rho}}\left[\left(\frac{\partial}{\partial\vec{\rho}}\log f(n)\right)\cdot\left(\frac{\partial}{\partial\vec{\rho}}\log f(n)\right)^{\dagger}\right]\\
 & = & \sum_{n=0}^{\infty}\frac{1}{f(n)}\left(\frac{\partial}{\partial\vec{\rho}}f(n)\right)\cdot\left(\frac{\partial}{\partial\vec{\rho}}f(n)\right)^{\dagger},
\end{eqnarray*}
where 
\[
\frac{\partial f}{\partial\rho_{ij}}=Q_{n}^{\beta}\left(\ket{\alpha_{i}}\bra{\alpha_{j}}\right).
\]

Notice that $\mathcal{I}(\vec{\rho})$ is a matrix-valued function
depending on the true state specified by $\vec{\rho}$. We use the
determinant of $\mathcal{I}(\vec{\rho})$ as a one-parameter measure
of the information contained in the measurement $Q_{n}^{\beta}(\rho)$
and plot $\det\mathcal{I}(\vec{\rho})$ as a function of $\beta$
for a few different $\vec{\rho}$ (see Fig. \ref{fig:fisher_info}). 

\begin{figure*}
\centering{}\includegraphics{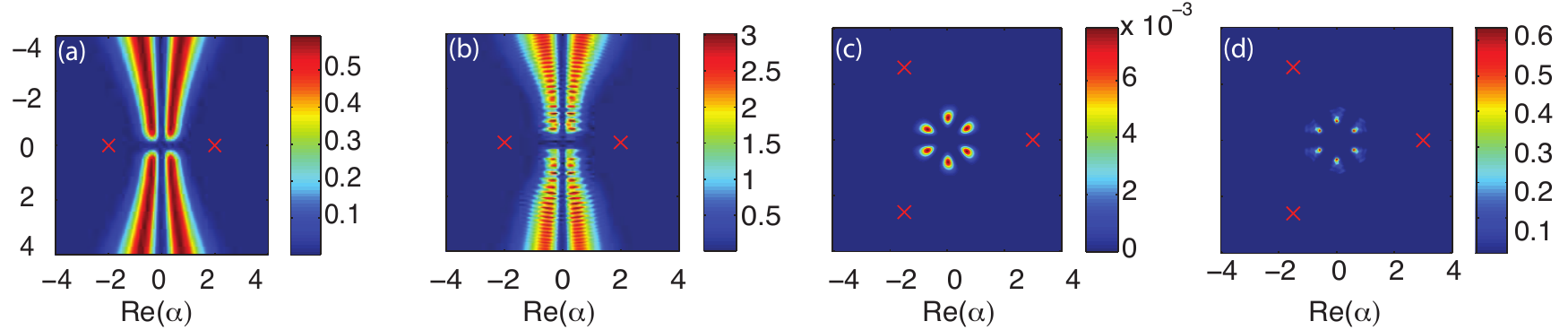}\protect\caption{Determinant of the Fisher information $\mathcal{I}(\vec{\rho})$ as
a function of $\beta$ for four different states. (a) Two-component
maximally mixed cat state, $\rho_{ij}\propto\delta_{ij}$. In other
words, the Bloch vector for the effective two level system is $\vec{0}$.
(b) A two-component cat state, with Bloch vector $0.9\cdot(1,\,1,\,0)/\sqrt{2}$.
(c) Three-component maximally mixed cat state, $\rho_{ij}\propto\delta_{ij}$.
(d) A mixture $\rho=(1-\lambda)I/3+\lambda\protect\ket{\psi}\protect\bra{\psi}$
where $I$ is the identity and $\protect\ket{\psi}=(1,\,1,\,1)^{\dagger}\sqrt{3}$.
The shape of the good detection region for maximally mixed states
is very similar to that predicted by the condition number while for
higher purity states additional ``interference fringes'' appear.
The worst case of Fisher information over all true states appears
to be that of the maximally mixed states. The good regions for $\beta$
predicted by worst-case Fisher information agree well with that given
by condition number. \label{fig:fisher_info}}
\end{figure*}
 This justifies the use of CN as a guide for optimizing measurement
schemes, which is much easier to calculate than the worst-case Fisher
information. For larger $p$ or general states, we need to consider
multiple measurement settings and optimized choices of $\beta's$
as discussed below.

\section{Informational completeness for general states \label{sec:gen-states-IC}}

We now consider general states with no structure other than an excitation
number cutoff $m_{c}$. To achieve IC, we need $N_{\beta}=(m_{c}+1)$
different $\beta's$ as argued below. In the Fock basis, \textcolor{black}{$\rho=\sum_{m_{1},m_{2}=0}^{m_{c}}\rho_{m_{1},m_{2}}\ket{m_{1}}\bra{m_{2}}$,
and} for each term $\ket{m_{1}}\bra{m_{2}}$ 

\begin{eqnarray*}
 &  & Q_{n}^{\beta}\left(\ket{m_{1}}\bra{m_{2}}\right)\\
 & = & \frac{\left|\beta\right|^{2n}e^{-\left|\beta\right|^{2}}}{n!}\frac{\sqrt{m_{1}!m_{2}!}}{(-\beta)^{m_{1}}(-\beta^{*})^{m_{2}}}\mathcal{L}_{m_{1}}^{n-m_{1}}(\left|\beta\right|^{2})\mathcal{L}_{m_{2}}^{n-m_{2}}(\left|\beta\right|^{2}),
\end{eqnarray*}
where $\mathcal{L}_{m}^{n}(x)$ is the associated Laguerre polynomial.
Note that $\mathcal{L}_{m}^{n}(x)$ is not only a polynomial of degree
$m$ in $x$ but also a polynomial of degree $m$ in $n$. Apart from
the factor $\frac{\left|\beta\right|^{2n}e^{-\left|\beta\right|^{2}}}{n!}$,
$Q_{n}^{\beta}\left(\ket{m_{1}}\bra{m_{2}}\right)$ is a polynomial
of degree $(m_{1}+m_{2})$ in $n$. Since $Q_{n}^{\beta}(\rho)$ has
a degree of $2m_{c}$ in $n$, experimental values of of $Q_{n}^{\beta}(\rho)$
for each $\beta$ provide $(2m_{c}+1)$ real coefficients, 
\[
Q_{n}^{\beta}(\rho)=\sum_{k=0}^{2m_{c}}n^{k}\cdot c_{k}^{\beta}.
\]
 The dependence of $c_{k}^{\beta}$ on $\rho_{m_{1}m_{2}}$ is shown
below (omitting $\beta$ superscript on $c_{k}$ ):

\begin{eqnarray*}
c_{2m_{c}} & \sim & \rho_{m_{c},m_{c}}\\
c_{2m_{c}-1} & \sim & \rho_{m_{c},m_{c}},\rho_{m_{c}-1,m_{c}},\rho_{m_{c},m_{c}-1}\\
 & \vdots\\
c_{m_{c}} & \sim & \rho_{0,m_{c}},\rho_{1,m_{c}-1},\,\cdots,\rho_{m_{c},0}\ \text{and all above}\\
c_{m_{c}+1} & \sim & m_{c}\ \text{new terms and all above}\\
 & \vdots\\
c_{0} & \sim & \text{all variables above}.
\end{eqnarray*}
For example, knowledge of $c_{2m_{c}}$ directly reveals $\rho_{m_{c},m_{c}}$
and $c_{2m_{c}-1}$ gives a linear equation involving $\rho_{m_{c},m_{c}-1}$,
$\rho_{m_{c}-1,m_{c}}$ and $\rho_{m_{c},m_{c}}$ which is already
obtained from $c_{2m_{c}}$. After experimentally obtaining $Q_{n}^{\beta_{1}}(\rho)$
and $Q_{n}^{\beta_{2}}(\rho)$, the values of $\rho_{m_{c},m_{c}-1}$
and $\rho_{m_{c}-1,m_{c}}$ can be determined. Continuing this way
we can determine all of $\rho_{m_{1},m_{2}}$ after measuring $Q_{n}^{\beta}(\rho)$
for $(m_{c}+1)$ $\beta's$. This analysis is similar to that done
in \cite{info_completeness}.

\section{Error robustness for general states \label{sec:gen-state-robustness}}

It is convenient to consider the covariance matrix, 
\[
C\equiv A^{\dagger}A=\sum_{j}A_{\beta_{j}}^{\dagger}A_{\beta_{j}},
\]
and $\kappa(C)=\kappa(A)^{2}$. The element $C_{(m_{1}m_{2}),(n_{1}n_{2})}$
is the overlap of the columns of $A$ corresponding to $\ket{m_{1}}\bra{m_{2}}$
and $\ket{n_{1}}\bra{n_{2}}$. In the ideal case, where $\kappa(A)=1$
and $A$ is an isometry, $C$ should be proportional to the identity
matrix. Using 
\begin{eqnarray*}
A_{(n,\beta),(m_{1},m_{2})} & = & tr\left[D(\beta)\ket n\bra nD(-\beta)\ket{m_{1}}\bra{m_{2}}\right]\\
 & = & e^{-\left|\beta\right|^{2}}\frac{1}{n!}\left|\beta\right|^{2n}\frac{\sqrt{m_{1}!}}{(-\beta)^{m_{1}}}\mathcal{L}_{m_{1}}^{n-m_{1}}\left(\left|\beta\right|^{2}\right)\\
 &  & \times\frac{\sqrt{m_{2}!}}{\left[(-\beta)^{m_{2}}\right]^{*}}\mathcal{L}_{m_{2}}^{n-m_{2}}\left(\left|\beta\right|^{2}\right),
\end{eqnarray*}
we see that 
\[
A_{(n,\beta),(m_{1},m_{2})}\propto\beta^{m_{2}-m_{1}}g_{m_{1}m_{2}}(\left|\beta\right|),
\]
and 
\begin{eqnarray*}
C_{(m_{1}m_{2}),(n_{1}n_{2})} & = & \sum_{n,j}A_{(n,\beta_{j}),(m_{1},m_{2})}^{*}A_{(n,\beta_{j}),(n_{1},n_{2})}\\
 & \propto & \sum_{\beta_{j}}\beta_{j}^{m_{1}-m_{2}-n_{1}+n_{2}}f_{m_{1},m_{2},n_{1},n_{2}}(\left|\beta_{j}\right|),
\end{eqnarray*}
where $g$ and $f$ are real functions that do not have dependence
on the complex argument of $\beta's$. Note the convenient fact that
the matrix $C$ is additive for parts corresponding to different $\beta's$.
Consider a set of $\beta's$ with the same magnitude, $\beta_{j}=\left|\beta\right|e^{i\phi_{j}}$.
Partitioning the indices $(m_{1}m_{2})$ and $(n_{1}n_{2})$ into
groups according to $k_{1}\equiv m_{1}-m_{2}$ and $k_{2}\equiv n_{1}-n_{2}$,
$C$ has a block structure $C=\left[C_{k_{1}k_{2}}\right]$, where
elements of the block $C_{k_{1}k_{2}}$ are proportional to $\sum_{j}e^{-i(k_{1}-k_{2})\phi_{j}}$. 

Both intuitively and rigorously, eliminating the off-diagonal blocks
with $k_{1}\ne k_{2}$ would reduce the condition number. This is
also known as ``pinching'' in matrix analysis (see also Appendix
\ref{sec:appdx-rings}). We may use $N_{\beta}=(2m_{c}+1)$ measurement
settings with $\beta's$ evenly distributed over a circle with 
\[
\phi_{j}=\frac{2\pi}{2m_{c}+1}j,\ \text{for }j=0,\,1,\,\cdots,\,2m_{c},
\]
which is denoted as ``full-ring configuration'' or FRC, as shown
in the inset of Fig.~\ref{fig:betas_general_states}. As pointed
out in Appendix \ref{sec:appdx-rings}, the multiple-full-ring configuration
(MFRC) should be optimal. However we observed numerically \textcolor{red}{{}
}that the improvement of MFRC over the FRC with optimal ring radius
is extremely small or even zero. Denote the covariance matrix for
a ring of $(2m_{c}+1)$ $\beta's$ with radius $r$ as $C_{r}$. We
compared $\min_{r}\kappa(C_{r})$ and $\min_{r_{1},r_{2}}\kappa(C_{r_{1}}+C_{r_{2}})$.
For $m_{c}=1$ we found a 1.6\% difference and for $m_{c}\ge2$ (tested
up to 7) they are equal. We thus conjecture that FRC is the optimal
configuration for $m_{c}\ge2$. The number of $\beta's$ required
for MFRC is at least twice as large as that of FRC. So practically
FRC is much more efficient than MFRC. 

\begin{figure}
\centering{}\includegraphics{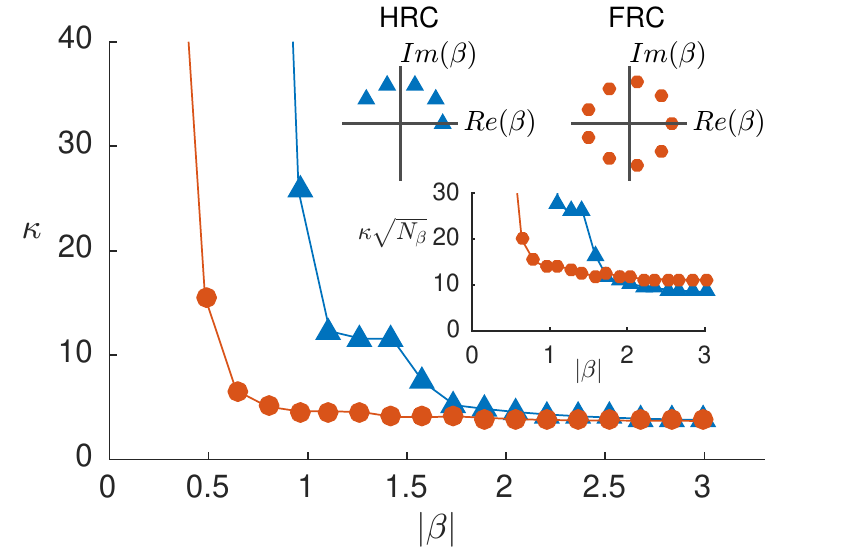}\protect\caption{(Color online) Main panel: Condition numbers of full-ring configuration
(FRC) and half-ring configuration (HRC) as a function of the ring
radius ($m_{c}=4$ case). Top two insets: FRC and HRC in phase space.
For both schemes, $\beta_{j}=\left|\beta\right|e^{i\phi_{j}}$. FRC:
$\phi_{j}=\frac{2\pi}{2m_{c}+1}j$, $j=0,\,1,\,2,\,\cdots,\,2m_{c}$.
HRC: $\phi_{j}=\frac{\pi}{2m_{c}+1}j$, $j=0,\,1,\,2,\,\cdots,\, m_{c}$.
The condition number of HRC approaches that of FRC as $\left|\beta\right|$
gets large, as predicted by theory. Bottom inset:\textbf{ }Figure
of merit $\kappa\sqrt{N_{\beta}}$ for HRC and FRC. \label{fig:betas_general_states}}
\end{figure}

Strictly speaking, with a smaller $N_{\beta}$ it is not possible
to fully pinch matrix $C$, i.e. satisfying 
\[
\sum_{j}e^{-i(k_{1}-k_{2})\phi_{j}}\propto\delta_{k_{1}k_{2}},
\]
for all $k_{1}$, $k_{2}$. This justifies the ring based configurations
used in \cite{qn_tomo_welsch,qn_tomo_Wineland,Martinis_qn_tomo_ring}.
Numerically, however, we find that for large $\left|\beta\right|$,
the number of measurement settings can be further reduced from $2m_{c}+1$
to $m_{c}+1$ without compromising CN, as illustrated in Fig.~\ref{fig:betas_general_states}.
The optimized $\beta$'s are evenly distributed over half a circle,
with 
\[
\phi_{j}=\frac{\pi}{m_{c}+1}j,\ \text{for }j=0,\,1,\,\cdots,\, m_{c},
\]
which is denoted as ``half-ring configuration'' or HRC, as shown
in the inset of Fig.~\ref{fig:betas_general_states}. For even $m_{c}$,
the configuration $\phi_{j}=\frac{2\pi}{m_{c}+1}j$, for $j=0,\,1,\,\cdots,\, m_{c}$,
works as well. The justification of HRC lies in the special asymptotic
behavior of matrix $C$. As $\left|\beta\right|$ gets large, the
off-diagonal blocks of $C_{k_{1},k_{2}}$ with odd $k_{1}-k_{2}$
scale as $1/\left|\beta\right|^{2}$, negligible compared to those
$C_{k_{1},k_{2}}$ with even $k_{1}-k_{2}$ which scales as $1/\left|\beta\right|$
(see Appendix \ref{sec:appdx-proof-thm} for a proof). So nearly half
of those off-diagonal blocks are automatically pinched and we only
need to have 
\[
\sum_{j}e^{-i(k_{1}-k_{2})\phi_{j}}\propto\delta_{k_{1}k_{2}},\ \text{for even }k_{1}-k_{2},
\]
to fully pinch $C$, which can be achieved using $m_{c}+1$ settings.
Interestingly, the pinching analysis can be applied to Homodyne detection
(see Appendix \ref{sec:appdx-homodyne}) and we verified that the
intuitive choice of equally spaced phase angles is optimal. Furthermore,
we found that the matrix $C$ for $Q_{n}$ asymptotes to that of Homodyne
detection and so Homodyne detection can in some sense be seen as the
$Q_{n}$ detection with $\beta\rightarrow\infty$. 

We also performed numerical gradient-based optimization of $\kappa(A)$
over $\beta's$ with different $N_{\beta}$. The gradient of CN with
respect to $\beta's$ can be calculated using perturbation theory
(detailed in Appendix \ref{sec:appdx-gradient}). CN drops significantly
as $N_{\beta}$ increases to $m_{c}+1$ and does not improve further
when $N_{\beta}>m_{c}+1$. For each $N_{\beta}$ we initialize the
optimization with a large number of different configurations of $\beta's$
and HRC turns out the best (with the exception of the case $m_{c}=1$).
As a function of $m_{c}$, the asymptotic CN grows slowly, $\kappa(A)\sim m_{c}^{1/2}$
(see Fig. \ref{fig:kappa_vs_mc_Qn}). 

\begin{figure}
\centering{}\includegraphics{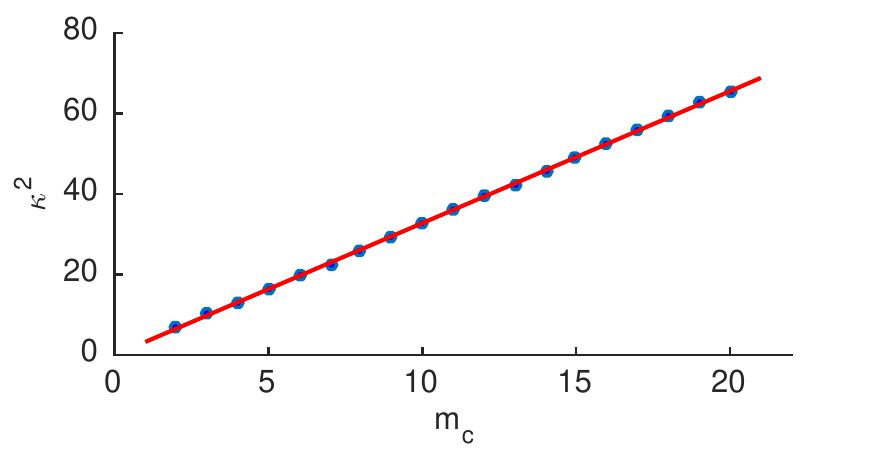}\protect\caption{Optimal condition number for $Q_{n}$ measurements as a function of
$m_{c}$. Vertical axis shows $\kappa(A)^{2}$. Red solid line shows
a linear fit with equation $\kappa^{2}=3.28m_{c}-0.07769$. \label{fig:kappa_vs_mc_Qn}}
\end{figure}

\section{Discussion of noise models \label{sec:noise-model}}

So far, we have assumed that $\left\Vert \delta\vec{b}\right\Vert _{2}/\left\Vert \vec{b}\right\Vert _{2}$
is fixed, and minimize $\kappa(A)$ to optimize the bound in Eq. (\ref{eq:error_bound}).
On the other hand,  $\left\Vert \delta\vec{b}\right\Vert _{2}/\left\Vert \vec{b}\right\Vert _{2}$
might be tunable. A practically relevant situation is shot noise,
with 
\[
\left\Vert \delta\vec{b}\right\Vert _{2}/\left\Vert \vec{b}\right\Vert _{2}\propto1/\sqrt{N_{rep}}.
\]
 Meanwhile, $\kappa(A)$ depends on the number of measurement settings
$N_{\beta}$. Given total number of measurements (or copies of unknown
states) $N_{tot}=N_{\beta}\cdot N_{rep}$, we need to minimize $\tilde{\epsilon}\equiv\kappa(A)\left\Vert \delta\vec{b}\right\Vert /\left\Vert \vec{b}\right\Vert $
to have a better bound. Hence, 
\[
\tilde{\epsilon}\propto\kappa(A)/\sqrt{N_{rep}}=\kappa(A)\sqrt{N_{\beta}/N_{tot}}
\]
 implies that we should minimize $\kappa(A)\sqrt{N_{\beta}}$. As
illustrated in the bottom inset of Fig.~\ref{fig:betas_general_states},
HRC has lower $\kappa(A)\sqrt{N_{\beta}}$ for large $\left|\beta\right|$,
and is more robust than FRC in that regime.  In terms of scaling
with $m_{c}$, 
\[
\kappa(A)\sqrt{N_{\beta}}\sim m_{c}^{1/2}\sqrt{m_{c}+1}\sim m_{c}
\]
 for HRC and FRC while $\kappa(A)\sqrt{N_{\beta}}$ appears super-linear
in $m_{c}$ for Wigner tomography, as shown in Fig. \ref{fig:merit_Qn_Wig}.
The relative advantage of $Q_{n}$ tomography grows as $m_{c}$ increases. 

\begin{figure}
\centering{}\includegraphics{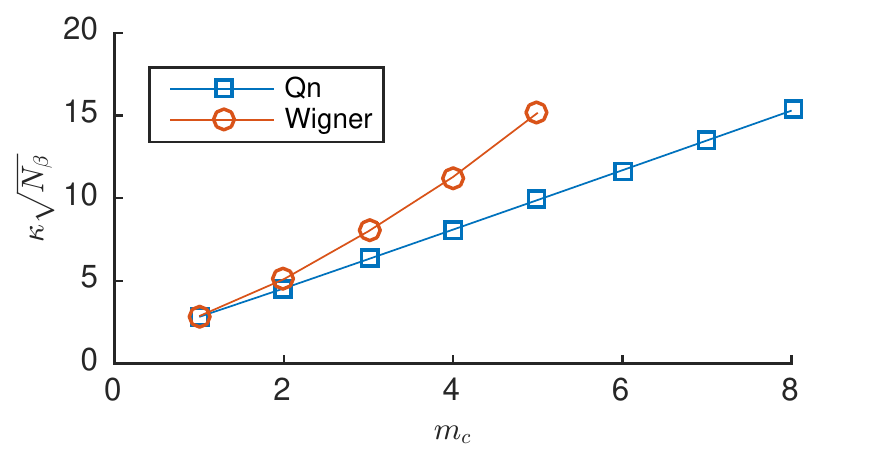}\protect\caption{Comparison of the figures of merits (assuming shot noise only) $\kappa\sqrt{N_{\beta}}$
for optimized $Q_{n}$ tomography with large enough $\left|\beta\right|$
and optimized Wigner tomography obtained from gradient-based optimization.
\label{fig:merit_Qn_Wig}}
\end{figure}

\begin{figure*}
\centering{}\includegraphics{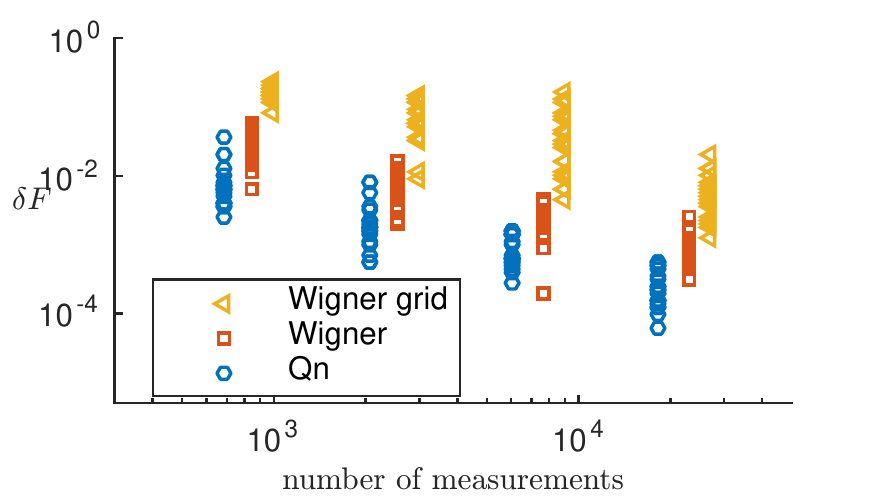}\includegraphics{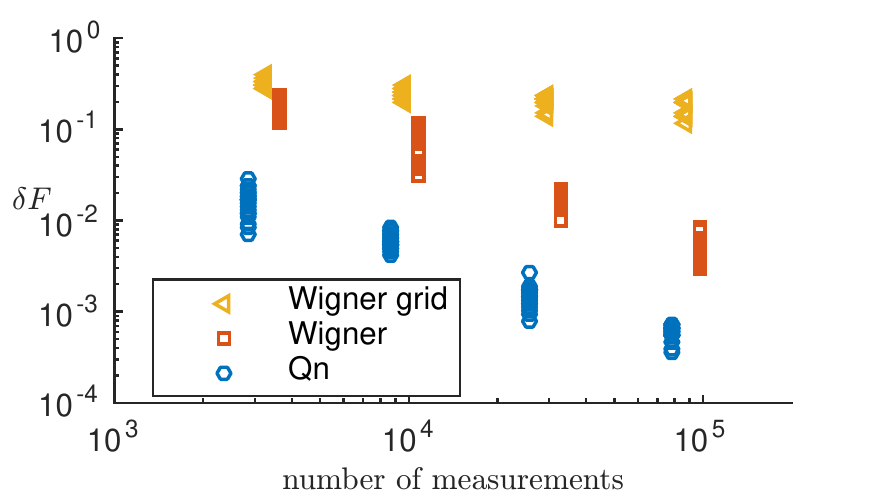}\protect\caption{(Color online) Comparison of performances of Wigner measurements where
$\beta's$ form a square lattice (\textcolor{black}{yellow} triangles),
Wigner measurements with optimized measurement settings obtained from
gradient search (red squares), and $Q_{n}$ measurements with optimized
measurement settings (blue circles). \textbf{Left/Right panels} correspond
to $m_{c}=2$ and $5$. The true state $\rho$ is a randomly generated
density matrix with excitation number cutoff $m_{c}=5$. Each scatter
point corresponds to one reconstruction via semi-definite programming
based on a set of simulated measurement records containing only shot
noise. The $y$ axis shows the reconstruction infidelity $\delta F=1-F(\rho,\,\rho')$
and the $x$ axis shows the total number of measurements performed,
i.e., total number of copies of unknown states consumed. \label{fig:err_vs_Ncp}}
\end{figure*}

\section{Benchmarking with Simulated Data\label{sec:benchmarking}}

Using simulated data (shot noise only), we tested and compared several
schemes, including Wigner measurements where $\beta's$ form a square
lattice (\textcolor{black}{yellow} triangles), Wigner measurements
with optimized $\beta's$ (red squares), and $Q_{n}$ measurements
with optimized $\beta's$ (blue circles). For each case reconstruction
is done by fitting a physical density matrix to the data, a semidefinite
program that can be solved efficiently with the Matlab package CVX
\cite{cvx,gb08}. Some typical results with $m_{c}=2$ and $5$ are
shown in Fig. \ref{fig:err_vs_Ncp}. Both optimized schemes have better
error scaling than the unoptimized one, because the bound for the
unoptimized case is too forgiving to suppress reconstruction error.
Between the two optimized schemes, the reconstruction infidelity for
the $Q_{n}$-based scheme is at least an order of magnitude smaller
than that of the Wigner-based scheme. Moreover, the advantage of using
$Q_{n}$ measurement and more generally optimized schemes indeed becomes
more significant for larger $m_{c}$, as predicted by the figure of
merit shown in Fig. \ref{fig:merit_Qn_Wig} and demonstrated by Fig.
\ref{fig:err_vs_Ncp}.

\section{Generalizations\label{sec:generalizations}}

The idea of optimizing the condition number of the measurement scheme
is completely general and can apply to the reconstruction problem
using \textcolor{black}{arbitrary} bases. Here we show one such example,
the generalized cat states, 
\[
\rho=\sum_{i,j,m_{1,}m_{2}}\rho_{i,m_{1};j,m_{2}}\ket{\alpha_{i},\, m_{1}}\bra{\alpha_{j},\, m_{2}},
\]
 where $i,\, j=1,\,2,\,\cdots,\, p$ and $m_{1},\, m_{2}=0,\,1,\,\cdots,\, m_{c}$,
and 
\[
\ket{\alpha_{i},\, m_{i}}\equiv D(\alpha_{i})\ket{m_{i}}
\]
 are displaced Fock states. Such states may arise when an ideal cat
 state is subject to experimental noise and each coherent-state component
is deformed. Now each column of the sensing matrix has the form 
\[
(d_{i}d_{j}e^{i\phi_{ij}})^{n}P(n),
\]
 where $P(n)$ is a polynomial coming from the associated Laguerre
polynomials 
\[
P(n)=\mathcal{L}_{m_{1}}^{n-m_{1}}(\left|\beta\right|^{2})\mathcal{L}_{m_{2}}^{n-m_{2}}(\left|\beta\right|^{2}).
\]
 On a large scale of $n$, the change of $(d_{i}d_{j}e^{i\phi_{ij}})^{n}P(n)$
as a function of $n$ is dominated by the exponential part $(d_{i}d_{j}e^{i\phi_{ij}})^{n}$.
So just as in the cat state case the columns with distinct $d_{i}d_{j}e^{i\phi_{ij}}$
are linearly independent. For the $(m_{c}+1)^{2}$ columns that share
the same $d_{i}d_{j}e^{i\phi_{ij}}$ but different polynomials $P(n)$,
we need $(m_{c}+1)$ different $\beta$'s to completely fix all unknowns
as discussed previously. We can then run numerical optimization for
all $N\ge(m_{c}+1)$ and pick the optimal $N$. 

A simultaneous optimization of many $\beta's$ can often get stuck
in shallow local minima. Here we show an alternative greedy policy
for optimization that works pretty well, where we pick one best $\beta$
at a time. The procedure is as follows. 

(1) Start with an empty set $S=\emptyset$ of $\beta's$, keeping
all the $\alpha's$ but set $m_{c}=0$, which allows the condition
number to be finite with one $\beta$.

(2) Pick the optimal $\beta$ (in the sense that it combined with
those $\beta's$ in $S$ produces the lowest condition number) and
add it to the set $S$.

(3) If the optimal condition number is small enough, increase $m_{c}$
by one (otherwise keep it the same).

(4) Repeat steps (2) and (3) until one reaches the desired $m_{c}$. 

We give one example here for which the condition number as a function
of the next $\beta$ to pick is shown in Fig. \ref{fig:gen_cat_reconstruction}. 

\begin{figure*}
\centering{}\includegraphics[scale=1.2]{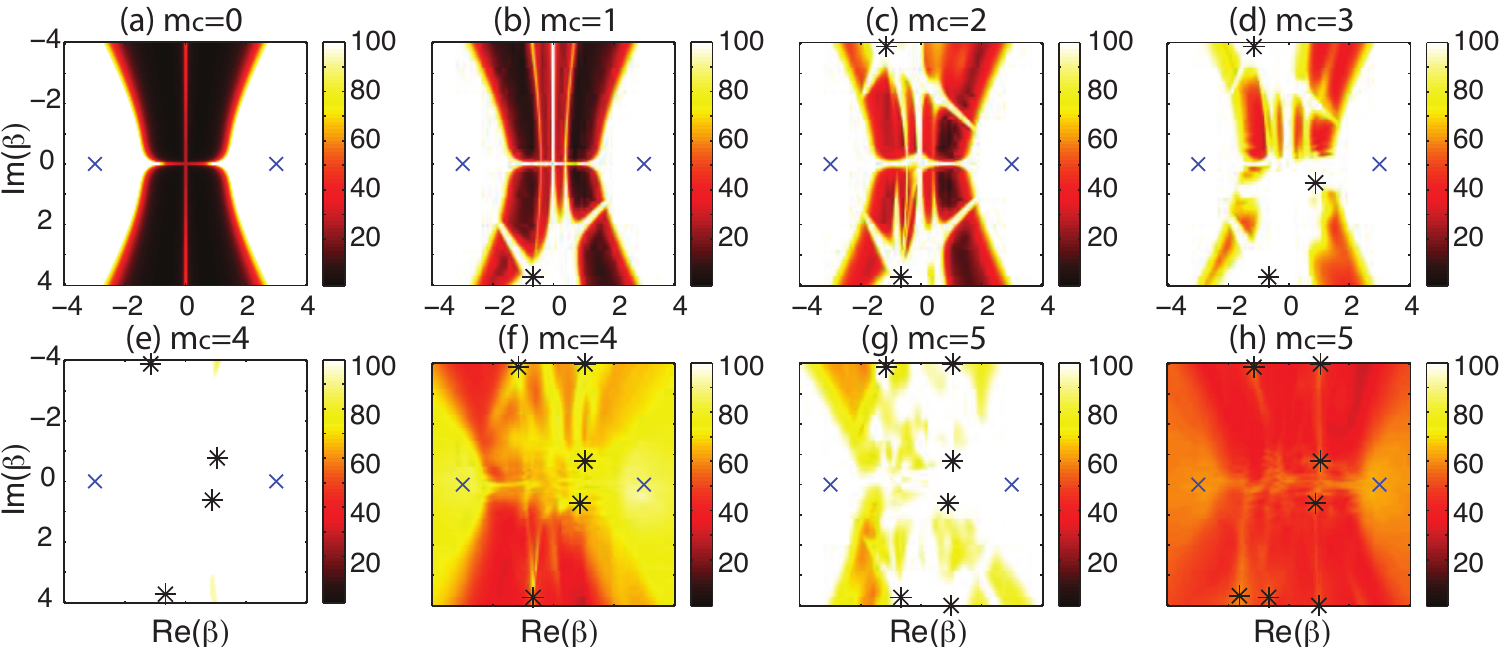}\protect\caption{Greedy optimization of the set of $\beta's$. Crosses show the position
of $\alpha's$ and stars indicate all the $\beta's$ added to the
set $S$. At each step, the optimal $\beta$ is added to the set $S$.
When the condition number is low enough (smaller than a preset threshold),
$m_{c}$ is increased by one and the optimization goes on. \label{fig:gen_cat_reconstruction}}
\end{figure*}

\section{Conclusion\label{sec:Conclusion}}

We proposed and analyzed a continuous variable QST scheme with the
full distribution information of excitation number after a variable
displacement. We showed how to construct a set of measurements that
has a small reconstruction error bound by optimizing a figure of merit
based on the condition number of the sensing map. For general states
with a given excitation number cutoff, we obtained the optimal displacement
patterns (half-ring and full-ring) that rationalize and improve the
previously considered ring-based choices. The idea of gradient-based
optimization of the condition number of the sensing map is versatile
and can apply to states expanded in an arbitrary basis and detection
methods that are parameterized by some continuous variables. As future
work, it is interesting to generalize the current scheme to QST for
multiple oscillators, spin ensembles \cite{Chen_Vuletic_2015}, and
CV process tomography. 
\begin{acknowledgments}
We acknowledge support from the ARL-CDQI, ARO (Grants No. W911NF-14-1-0011
and No. W911NF-14-1-0563), AFOSR MURI (Grants No. FA9550-14-1-0052
and No FA9550-14-1-0015), Alfred P. Sloan Foundation (Grant No. BR2013-049),
and Packard Foundation (Grant No. 2013-39273). We thank Siddharth
Prabhu and Jianming Wen for fruitful discussions.
\end{acknowledgments}

\appendix
\begin{widetext}

\section{Reconstructing A Physical Density Matrix \label{sec:appdx-physical-rho}}

Let $\vec{\rho'}$ be the least-squares solution (potentially non-physical)
from the noisy measurement record, 
\[
\vec{\rho'}=\left(A^{\dagger}A\right)^{-1}A^{\dagger}\left(\vec{b}+\delta\vec{b}\right).
\]
We claim that the physical density matrix $\tau$ that is closest
to $\rho'$ in the sense of some norm (say, the Frobenius norm) can
only be a better estimate of the true state $\rho$, i.e.,
\begin{equation}
\norm{\tau-\rho}_{F}\leq\norm{\rho'-\rho}_{F}.\label{eq:tau-rho-rho-prime}
\end{equation}

We now prove the above equation by contradiction. Suppose $\norm{\tau-\rho}_{F}>\norm{\rho'-\rho}_{F}$.
Now consider the triangle whose vertices are $\rho$, $\rho'$ and
$\tau$. Let $\theta\in[0,\pi]$ be the angle at the vertex $\tau$.
Using the Law of Cosines, we have that 
\[
\cos\theta=\frac{\norm{\rho'-\tau}_{F}^{2}-\norm{\rho-\rho'}_{F}^{2}+\norm{\rho-\tau}_{F}^{2}}{2\norm{\rho'-\tau}_{F}\norm{\rho-\tau}_{F}}>0.
\]
 This implies that $0\leq\theta<\pi/2$, i.e., the angle at $\tau$
is less than 90 deg. 

Hence, there exists a point $\zeta$ that is a convex combination
of $\tau$ and $\rho$ such that 
\[
\norm{\zeta-\rho'}_{F}<\norm{\tau-\rho'}_{F}.
\]
 Moreover, since $\rho$ and $\tau$ are physical density matrices
and the space of density matrices is convex, it follows that $\zeta$
is also physical. This contradicts the assumption that ``$\tau$
is the physical density matrix $\tau$ that is closest to $\rho'$.''
Therefore, we conclude that Eq. (\ref{eq:tau-rho-rho-prime}) must
hold.

Practically, $\tau$ can be obtained as the solution of the following
semidefinite program (SDP), 

\begin{eqnarray*}
\text{minimize} &  & \norm{\sigma-\rho'}_{F}\\
\text{subject to} & \  & \sigma\succeq0,\ \text{tr}\sigma=1.
\end{eqnarray*}
Note that SDP can be solved efficiently using the Matlab package CVX
\cite{cvx,gb08}. 

Alternatively, a physical reconstruction $\tau'$ may be obtained
by directly solving the least-squares problem in the space of physical
density matrices, i.e., 

\begin{eqnarray*}
\text{minimize} &  & \norm{A\cdot\vec{\sigma}-\vec{b}'}_{2}\\
\text{subject to} & \  & \sigma\succeq0,\ \text{tr}\sigma=1.
\end{eqnarray*}

\section{Bound for Reconstruction Error \label{sec:appdx-error-bound}}

We derive the  lower bound on the fidelity of reconstruction in terms
of condition number here. We will first find an upper bound for the
trace distance of the reconstructed state to the true state, and then
get the fidelity bound using the relation between fidelity and trace
distance $D(\rho,\,\sigma)$, 
\[
F(\rho,\,\sigma)\ge1-D(\rho,\,\sigma)
\]
 where $D(\rho,\,\sigma)=\frac{1}{2}\left\Vert \rho-\sigma\right\Vert _{\text{tr}}$. 

Let $\vec{\rho}$ be the true state and $\vec{\rho'}$ be the least-squares
solution from the noisy measurement record, 
\begin{eqnarray*}
\vec{\rho} & = & \left(A^{\dagger}A\right)^{-1}A^{\dagger}\vec{b},\\
\vec{\rho'} & = & \left(A^{\dagger}A\right)^{-1}A^{\dagger}\left(\vec{b}+\delta\vec{b}\right),
\end{eqnarray*}
and define $\delta\vec{\rho}\equiv\vec{\rho}-\vec{\rho'}=\tilde{A}^{-1}\delta\vec{b}=\left(A^{\dagger}A\right)^{-1}A^{\dagger}\delta\vec{b}$. 

Following the main text we use the two-norm for vectors $\vec{\rho}$
to define the condition number, then

\[
\left(\frac{\left\Vert \delta\vec{\rho}\right\Vert _{2}}{\left\Vert \vec{\rho}\right\Vert _{2}}\right)/\left(\frac{\left\Vert \delta\vec{b}\right\Vert _{2}}{\left\Vert \vec{b}\right\Vert _{2}}\right)\le\kappa(A).
\]
 Since the Frobenius norm of a matrix is the same as the two-norm
of it when arranged as a vector,
\[
\left\Vert \rho\right\Vert _{F}=\left\Vert \vec{\rho}\right\Vert _{2}\le\kappa(A)\left\Vert \rho\right\Vert _{F}\frac{\left\Vert \delta\vec{b}\right\Vert _{2}}{\left\Vert \vec{b}\right\Vert _{2}}.
\]

Let $\tau$ be the physical density matrix that best satisfies the
noisy measurement record $A\tau=\vec{b}+\delta\vec{b}$, obtained
as described in the previous section. We have 
\[
\left\Vert \rho-\tau\right\Vert _{\text{F}}\le\left\Vert \rho-\rho'\right\Vert _{\text{F}}=\left\Vert \delta\rho\right\Vert _{F}\le\kappa(A)\left\Vert \rho\right\Vert _{F}\frac{\left\Vert \delta\vec{b}\right\Vert _{2}}{\left\Vert \vec{b}\right\Vert _{2}},
\]
where the first inequality uses Eq. (\ref{eq:tau-rho-rho-prime}).
The above bound is useful since it upper bounds the distance (in terms
of the Frobenius norm) between the reconstructed state and the true
state. 

Using the relation between the trace norm and Frobenius norm 
\[
\left\Vert M\right\Vert _{\text{tr}}\le\sqrt{r}\left\Vert M\right\Vert _{F},
\]
we find 
\[
D(\rho,\,\tau)\le\frac{1}{2}\sqrt{r}\left\Vert \rho-\tau\right\Vert _{\text{F}}\le\frac{1}{2}\sqrt{r}\kappa(A)\left\Vert \rho\right\Vert _{F}\frac{\left\Vert \delta\vec{b}\right\Vert _{2}}{\left\Vert \vec{b}\right\Vert _{2}}
\]
 and 
\begin{eqnarray}
F(\rho,\,\tau) & \ge & 1-D(\rho,\,\tau)\ge1-\frac{1}{2}\sqrt{r}\kappa(A)\left\Vert \rho\right\Vert _{F}\frac{\left\Vert \delta\vec{b}\right\Vert _{2}}{\left\Vert \vec{b}\right\Vert _{2}}.\label{eqn-f}
\end{eqnarray}
In practice we have an estimate for the measurement noise $\epsilon\sim\frac{\left\Vert \delta\vec{b}\right\Vert _{2}}{\left\Vert \vec{b}\right\Vert _{2}}$
and the truncation dimension $d$ upperbounds the rank $r$ of $\delta\rho$.
Since $\rho$ is unknown we replace it with the reconstructed $\tau$.
In this way an approximate bound on the fidelity can be calculated,
$F(\rho,\,\tau)\gtrsim1-\frac{1}{2}\epsilon\sqrt{d}\kappa(A)\left\Vert \tau\right\Vert _{F}.$

\section{Discussion of Full- And Half-Ring Configurations \label{sec:appdx-rings}}

\subsubsection*{The Pinching Inequality}

Mathematically, wiping out all the off-diagonal blocks is called ``pinching''
and is formally described as 
\[
C\mapsto\tilde{C}=\sum_{k}P_{k}CP_{k},
\]
where $P_{k}$ is the projector to the subspace corresponding to the
block $C_{kk}$. It is known that the eigenvalues of $\tilde{C}$
are majorized by those of $C$ (see p. 50 of \cite{matrix_analysis}),
i.e., $\sum_{i=1}^{k}\lambda_{i}^{\downarrow}(\tilde{C})\le\sum_{i=1}^{k}\lambda_{i}^{\downarrow}(C)$
for $k=1,\,2,\,\cdots,\, D$ and $\sum_{i=1}^{D}\lambda_{i}^{\downarrow}(\tilde{C})=\sum_{i=1}^{D}\lambda_{i}^{\downarrow}(C)$,
where $\lambda_{i}^{\downarrow}$ are the eigenvalues in descending
order and $D$ is the dimension of $C$ and $\tilde{C}$. This implies
that $\kappa(\tilde{C})\le\kappa(C)$. This fact can also be understood
in the language of quantum mechanics. View $\tilde{C}$ as a block-diagonal
Hamiltonian $H_{0}$ and $C-\tilde{C}$ as a perturbation$H_{1}$
coupling different subspaces of $H_{0}$. It is well known that energy
levels repel each other when coupled to each other. So the highest
 energy level gets higher and the lowest gets lower, with their ratio
being increased. 

This means that among the sets of $\beta's$ with the same magnitude,
the FRC can give the optimal CN.

\subsubsection*{Multiple full-ring configuration gives lowest condition number}

We now argue that the MFRC can give the minimal condition number if
we do not limit the number of measurement settings. %
{} Here is our two-step argument. 

(a) Given a candidate configuration $\{\beta_{i}\}$ distributed on
a ring, i.e., $\left|\beta_{i}\right|=r$, we can always decrease
CN by rearranging or adding $\beta's$ such that the configuration
becomes FRC, i.e., pinching the covariance matrix. 

(b) For any given candidate set $\{\beta_{i}\}$ distributed on different
rings, we can always decrease the condition number by rearranging
or adding $\beta's$ such that the configuration becomes a collection
of FRC (MFRC) to pinch the covariance matrix. 

Numerically we observed that usually one full-ring configuration is
as good as the multiple full-ring configuration, except the case with
$m_{c}=1$ where a 1.6\% difference between single-ring and double-ring
configurations is found.

\subsubsection*{Half-ring configuration approximates full-ring configuration well }

We find it possible to simplify FRC further. With less than $(2m_{c}+1)$
points, it is impossible to exactly satisfy $\sum_{j}e^{i(k_{1}-k_{2})\phi_{j}}=\delta_{k_{1}k_{2}}$
for all $k_{1}$, $k_{2}$. However, we find a very special asymptotic
behavior of the covariance matrix, as stated by the following theorem
(see Appendix \ref{sec:appdx-proof-thm} for the proof). 
\begin{thm}
\label{thm:asymptotic-Cov}The large-$\left|\beta\right|$ asymptotic
form of $C_{m_{1}m_{2},m_{3}m_{4}}(\beta)$ is

\[
C_{m_{1}m_{2},m_{3}m_{4}}(\beta)\sim\begin{cases}
g(m_{1},\, m_{2},\, m_{3},\, m_{4},\,\phi)/\left|\beta\right|, & \sum_{i=1}^{4}m_{i}\ \text{is even};\\
g(m_{1},\, m_{2},\, m_{3},\, m_{4},\,\phi)/\left|\beta\right|^{2}, & \sum_{i=1}^{4}m_{i}\ \text{is odd};
\end{cases}
\]
where $\phi$ is the complex angle of $\beta$. 
\end{thm}
This theorem effectively says that the elements of $C(\beta)$ have
a \textbf{``parity selection rule.''} 

So in the large $\left|\beta\right|$ limit, the block $C_{k_{1}k_{2}}\sim1/\left|\beta\right|$
if $k_{1}-k_{2}$ is even and $C_{k_{1}k_{2}}\sim1/\left|\beta\right|^{2}$
if $k_{1}-k_{2}$ is odd. \textcolor{black}{Certainly, all diagonal
blocks $C_{kk}\sim1/\left|\beta\right|$.} So if $\left|\beta\right|$
is large enough, the blocks with odd $(k_{1}-k_{2})$ automatically
vanish. To make the rest of the off-diagonal blocks vanish, we only
need to choose a configuration such that $\sum_{j}e^{i(k_{1}-k_{2})\phi_{j}}=\delta_{k_{1}k_{2}}$
holds for even $k_{1}-k_{2}=2l$, where $l=0,\,\pm1,\,\pm2,\,\cdots,\,\pm m_{c}$,
i.e.,
\[
\sum_{j}e^{2il\phi_{j}}=\delta_{l,0.}
\]
It is straightforward to check that the HRC, $\phi_{j}=\frac{\pi}{m_{c}+1}j$
qualifies for all $m_{c}$ and $\phi_{j}=\frac{2\pi}{m_{c}+1}j$ qualifies
for even $m_{c}$. In fact for even $m_{c}$, $\phi_{j}=\frac{2\pi n}{m_{c}+1}j$
could work for any non-zero integer $n$. Therefore if the optimal
radius of FRC is large (which as we will show is usually the case),
HRC should work equally well with only half of the measurements.

\section{Optimal Setting For Homodyne Measurement \label{sec:appdx-homodyne}}

The pinching analysis to Homodyne tomography follows the $Q_{n}$
case closely. The term $\ket{m_{1}}\bra{m_{2}}$ contributes the Homodyne
signal 

\begin{eqnarray*}
\mathcal{H}(\ket{m_{1}}\bra{m_{2}}) & = & \trace{\left[\ket{x_{\theta}}\bra{x_{\theta}}\cdot\ket{m_{1}}\bra{m_{2}}\right]}\\
 & = & \frac{e^{i(m_{1}-m_{2})\theta}}{\pi^{1/2}\sqrt{2^{m_{1}+m_{2}}m_{1}!m_{2}!}}e^{-x^{2}}H_{m_{1}}(x)H_{m_{2}}(x).
\end{eqnarray*}
And the covariance matrix is 
\begin{eqnarray*}
C_{m_{1}m_{2},m_{3}m_{4}} & = & \frac{e^{i(m_{3}-m_{4}-m_{1}+m_{2})\theta}}{\pi\sqrt{2^{m_{1}+m_{2}+m_{3}+m_{4}}m_{1}!m_{2}!m_{3}!m_{4}!}}\int_{-\infty}^{+\infty}\, e^{-2x^{2}}H_{m_{1}}(x)H_{m_{2}}(x)H_{m_{3}}(x)H_{m_{4}}(x)\\
 & \equiv & \frac{e^{i(m_{3}-m_{4}-m_{1}+m_{2})\theta}}{\pi\sqrt{2^{m_{1}+m_{2}+m_{3}+m_{4}}m_{1}!m_{2}!m_{3}!m_{4}!}}g(m_{1},m_{2},m_{3},m_{4}).
\end{eqnarray*}

Due to the properties of the Hermite polynomials, i.e., $H_{n}(x)$
is an odd or even function of $x$ if $n$ is odd or even. If $m_{1}+m_{2}+m_{3}+m_{4}$
is odd, the integral 
\[
\int_{-\infty}^{+\infty}dx\, e^{-2x^{2}}H_{m_{1}}(x)H_{m_{2}}(x)H_{m_{3}}(x)H_{m_{4}}(x)=0.
\]
To pinch the covariance matrix, we can use the half-ring configuration,
i.e., pick $(m_{c}+1)$ $\theta_{j}$ such that $\theta_{j}=\frac{\pi}{2m_{c}+1}j$
where $j=0,\,1,\,2,\,\cdots,\, m_{c}$.  

Plugging definite values for $m_{1}$, $m_{2}$, $m_{3}$, $m_{4}$,
we find the covariance matrix for Homodyne to be the same (up to a
global  constant) as the asymptotic covariance matrix for $Q_{n}$
measurements.

\section{Numerical Calculation of the Gradient of the Condition Number \label{sec:appdx-gradient}}

We briefly outline how to calculate the gradient of a matrix's condition
number using perturbation theory, in the context of the state tomography
problem. 

Let us perturbe matrix $A$ by changing $\beta_{i}$ infinitesimally,
\begin{eqnarray*}
A(\beta_{i}+\delta\beta_{i}) & = & A+\delta\beta_{i}(\partial_{\beta_{i}}A)\\
 & \equiv & A+\delta\beta_{i}B_{i},
\end{eqnarray*}
 where matrix $B_{i}$ can be calculated from the explicit expression
of $A$. Note that we are changing only one $\beta_{i}$ so there
is no summation over $i$ here. We try to find $\partial_{\beta_{i}}\kappa(A)$.
For convenience we choose to work with the Hermitian covariance matrix
$C\equiv A^{\dagger}A$ whose condition number is $\kappa(C)=\kappa(A^{\dagger}A)=\kappa(A)^{2}$.
\begin{eqnarray}
\partial_{\beta_{i}}\kappa(C) & = & \partial_{\beta_{i}}\frac{\epsilon_{\text{max}}(C)}{\epsilon_{\text{min}}(C)}\nonumber \\
 & = & \frac{\partial_{\beta_{i}}\epsilon_{\text{max}}(C)\epsilon_{\text{min}}(C)-\epsilon_{\text{max}}(C)\partial_{\beta_{i}}\epsilon_{\text{min}}(C)}{\epsilon_{\text{min}}(C)^{2}},\label{eq:grad_cond_number}
\end{eqnarray}
where $\epsilon_{\text{max}}/\epsilon_{\text{min}}$ are the largest/smallest
eigenvalues of $C$. Now the problem reduces to calculate the gradient
of the eigenvalues of $C$ with respect to $\beta_{i}$. 

It is well known in quantum mechanics that the first-order perturbation
to the energy of the $k$th eigenstate is
\[
\delta\epsilon_{k}=\bra{\psi_{k}}\delta H\ket{\psi_{k}}
\]
 where $\ket{\psi_{k}}$ is the $k$th eigenstate of the unperturbed
Hamiltonian $H$ and $\delta H$ is a small perturbation. 

In our case, 
\[
C(\beta_{i}+\delta\beta_{i})=C+\delta\beta_{i}(B_{i}^{\dagger}A+A^{\dagger}B_{i})+O(\delta\beta^{2}),
\]
 so 
\begin{equation}
\partial_{\beta_{i}}\epsilon_{\text{k}}(C)=v_{k}^{\dagger}(B_{i}^{\dagger}A+A^{\dagger}B_{i})v_{k},\label{eq:grad_eigenvalue}
\end{equation}
 where $v_{k}$ is the $k$th eigenvector of $C$.

\section{Proof of Theorem \ref{thm:asymptotic-Cov} \label{sec:appdx-proof-thm}}

For completeness, we provide the detailed proof of theorem \ref{thm:asymptotic-Cov}
in this appendix.

\subsection*{Some Preparation }
\begin{lem}
\label{lem:BesseI-derivative-iterations}Let $I_{\nu}(z)$ denote
the modified Bessel functions of the first kind. For any non-negative
integer $k$, we have 
\begin{eqnarray*}
\frac{\partial^{k}}{\partial z^{k}}\left[(2\sqrt{z})^{\nu}I_{\nu}(2\sqrt{z})\right] & =2^{k} & \left[(2\sqrt{z})^{\nu-k}I_{\nu-k}(2\sqrt{z})\right],\\
\frac{\partial^{k}}{\partial z^{k}}\left[(2\sqrt{z})^{-\nu}I_{\nu}(2\sqrt{z})\right] & =2^{k} & \left[(2\sqrt{z})^{-(\nu+k)}I_{\nu+k}(2\sqrt{z})\right],
\end{eqnarray*}
\end{lem}
\begin{proof}
These can be verified using the properties of $I_{\nu}(z)$. 
\end{proof}

\begin{lem}
\label{lem:bino-prod-sum}Let $n$, $j_{1}$, $j_{2}$, $j_{3}$,
$j_{4}$ be non-negative integers, %
then we have 
\begin{eqnarray*}
 &  & \sum_{n=0}^{\infty}\frac{z^{n}}{(n!)^{2}}\left(\begin{array}{c}
n\\
j_{1}
\end{array}\right)\left(\begin{array}{c}
n\\
j_{2}
\end{array}\right)\left(\begin{array}{c}
n\\
j_{3}
\end{array}\right)\left(\begin{array}{c}
n\\
j_{4}
\end{array}\right)\\
 & = & \frac{1}{j_{1}!j_{2}!j_{3}!j_{4}!}z^{j_{4}}\frac{\partial^{j_{4}}}{\partial z^{j_{4}}}z^{j_{3}}\frac{\partial^{j_{3}}}{\partial z^{j_{3}}}z^{j_{2}}\frac{\partial^{j_{2}}}{\partial z^{j_{2}}}z^{j_{1}}\frac{\partial^{j_{1}}}{\partial z^{j_{1}}}I_{0}(2\sqrt{z})\\
 & = & \sum_{k_{4}=0}^{j_{4}}\sum_{k_{3}=0}^{j_{3}}\frac{j_{4}!}{k_{4}!(j_{4}-k_{4})!}\frac{j_{3}!}{k_{3}!(j_{3}-k_{3})!}\frac{j_{2}!}{(j_{2}-k_{3})!}\frac{(j_{2}+j_{3}-k_{3})!}{(j_{2}+j_{3}-k_{3}-k_{4})!}(\sqrt{z})^{j_{1}+j_{2}+j_{3}+j_{4}-k_{3}-k_{4}}I_{j_{1}-j_{2}-j_{3}-j_{4}+k_{3}+k_{4}}(2\sqrt{z}).
\end{eqnarray*}
where $I_{0}(2\sqrt{z})=\sum_{n=0}^{\infty}\frac{z^{n}}{(n!)^{2}}$. \end{lem}
\begin{proof}
It is straightforward to show tha
\begin{eqnarray*}
\sum_{n}\frac{z^{n}}{(n!)^{2}}\left(\begin{array}{c}
n\\
j
\end{array}\right) & = & \frac{1}{j!}\sum_{n}\frac{z^{n}}{(n!)^{2}}n(n-1)\cdots(n-j+1)=\frac{1}{j!}z^{j}\frac{\partial^{j}}{\partial z^{j}}I_{0}(2\sqrt{z}).
\end{eqnarray*}
Similarly,
\[
\sum_{n=0}^{\infty}\frac{z^{n}}{(n!)^{2}}\left(\begin{array}{c}
n\\
j_{1}
\end{array}\right)\left(\begin{array}{c}
n\\
j_{2}
\end{array}\right)\left(\begin{array}{c}
n\\
j_{3}
\end{array}\right)\left(\begin{array}{c}
n\\
j_{4}
\end{array}\right)=\frac{1}{j_{1}!j_{2}!j_{3}!j_{4}!}z^{j_{4}}\frac{\partial^{j_{4}}}{\partial z^{j_{4}}}z^{j_{3}}\frac{\partial^{j_{3}}}{\partial z^{j_{3}}}z^{j_{2}}\frac{\partial^{j_{2}}}{\partial z^{j_{2}}}z^{j_{1}}\frac{\partial^{j_{1}}}{\partial z^{j_{1}}}I_{0}(2\sqrt{z}).
\]
\begin{eqnarray*}
\end{eqnarray*}
We now try to express the above quantity in an explicit form. 

First, using Lemma \ref{lem:BesseI-derivative-iterations},
\begin{eqnarray*}
\frac{\partial^{j_{1}}}{\partial z^{j_{1}}}I_{0}(2\sqrt{z}) & = & 2^{j_{1}}\left[(2\sqrt{z})^{-j_{1}}I_{-j_{1}}(2\sqrt{z})\right].
\end{eqnarray*}

Next, 
\begin{eqnarray*}
\frac{\partial^{j_{2}}}{\partial z^{j_{2}}}z^{j_{1}}\frac{\partial^{j_{1}}}{\partial z^{j_{1}}}I_{0}(2\sqrt{z}) & = & 2^{j_{2}-j_{1}}\left[(2\sqrt{z})^{j_{1}-j_{2}}I_{j_{1}-j_{2}}(2\sqrt{z})\right].
\end{eqnarray*}

Continuing this, we can get
\begin{eqnarray*}
\frac{\partial^{j_{3}}}{\partial z^{j_{3}}}z^{j_{2}}\frac{\partial^{j_{2}}}{\partial z^{j_{2}}}z^{j_{1}}\frac{\partial^{j_{1}}}{\partial z^{j_{1}}}I_{0}(2\sqrt{z}) & = & 2^{j_{2}-j_{1}}\frac{\partial^{j_{3}}}{\partial z^{j_{3}}}z^{j_{2}}\left[(2\sqrt{z})^{j_{1}-j_{2}}I_{j_{1}-j_{2}}(2\sqrt{z})\right]\\
 & = & 2^{j_{2}-j_{1}}\sum_{k_{3}=0}^{j_{3}}\left(\begin{array}{c}
j_{3}\\
k_{3}
\end{array}\right)\frac{\partial^{k_{3}}}{\partial z^{k_{3}}}(z^{j_{2}})\frac{\partial^{j_{3}-k_{3}}}{\partial z^{j_{3}-k_{3}}}\left[(2\sqrt{z})^{j_{1}-j_{2}}I_{j_{1}-j_{2}}(2\sqrt{z})\right]\\
 & = & 2^{j_{3}+j_{2}-j_{1}}\sum_{k3=0}^{j_{3}}\left(\begin{array}{c}
j_{3}\\
k_{3}
\end{array}\right)\frac{\partial^{k_{3}}}{\partial z^{k_{3}}}(z^{j_{2}})2^{-k_{3}}\left[(2\sqrt{z})^{j_{1}-j_{2}-j_{3}+k_{3}}I_{j_{1}-j_{2}-j_{3}+k_{3}}(2\sqrt{z})\right]
\end{eqnarray*}
Eventually we obtain 
\begin{eqnarray*}
 &  & z^{j_{4}}\frac{\partial^{j_{4}}}{\partial z^{j_{4}}}z^{j_{3}}\frac{\partial^{j_{3}}}{\partial z^{j_{3}}}z^{j_{2}}\frac{\partial^{j_{2}}}{\partial z^{j_{2}}}z^{j_{1}}\frac{\partial^{j_{1}}}{\partial z^{j_{1}}}I_{0}(2\sqrt{z})\\
 & = & \sum_{k_{4}=0}^{j_{4}}\sum_{k_{3}=0}^{j_{3}}\frac{j_{4}!}{k_{4}!(j_{4}-k_{4})!}\frac{j_{3}!}{k_{3}!(j_{3}-k_{3})!}\frac{j_{2}!}{(j_{2}-k_{3})!}\frac{(j_{2}+j_{3}-k_{3})!}{(j_{2}+j_{3}-k_{3}-k_{4})!}(\sqrt{z})^{j_{1}+j_{2}+j_{3}+j_{4}-k_{3}-k_{4}}I_{j_{1}-j_{2}-j_{3}-j_{4}+k_{3}+k_{4}}(2\sqrt{z}).
\end{eqnarray*}
Note that in the above derivation, factors like 
\[
\frac{a!}{(a-b)!}=a(a-1)(a-2)\cdots(a-b+1)
\]
are naturally interpreted as zero if $a<b$. 
\end{proof}

\begin{lem}
\label{lem:moments-binom}Let $m$ be a positive integer and $k$
is a non-negative integer, 
\[
\sum_{i=0}^{m}(-1)^{i}\left(\begin{array}{c}
m\\
i
\end{array}\right)i^{k}=\begin{cases}
0, & \text{if }0\le k<m;\\
(-1)^{m}m!, & \text{if }k=m;\\
(-1)^{m}m!\left(\begin{array}{c}
m+1\\
2
\end{array}\right), & \text{if }k=m+1.
\end{cases}
\]
 
\end{lem}

\begin{proof}
Let $\alpha$ be any real number,
\begin{eqnarray*}
(1+\alpha)^{m} & = & \sum_{i=0}^{m}\alpha^{i}\left(\begin{array}{c}
m\\
i
\end{array}\right).
\end{eqnarray*}
We then have 

\[
(\alpha\frac{\partial}{\partial\alpha})^{k}(1+\alpha)^{m}=\sum_{i=0}^{m}i^{k}\alpha^{i}\left(\begin{array}{c}
m\\
i
\end{array}\right).
\]

Defining $x\equiv1+\alpha$, we have 
\[
\alpha\frac{\partial}{\partial\alpha}=(\alpha+1-1)\frac{\partial}{\partial(\alpha+1)}=(x-1)\frac{\partial}{\partial x}=x\frac{\partial}{\partial x}-\frac{\partial}{\partial x}.
\]
So
\begin{eqnarray}
\sum_{i=0}^{m}(-1)^{i}\left(\begin{array}{c}
m\\
i
\end{array}\right)i^{k} & = & \left.(\alpha\frac{\partial}{\partial\alpha})^{k}(1+\alpha)^{m}\right|_{\alpha=-1}\nonumber \\
 & = & \left.\left(x\frac{\partial}{\partial x}-\frac{\partial}{\partial x}\right)^{k}x^{m}\right|_{x=0}.\label{eq:lemma_important_eq}
\end{eqnarray}
Expanding $\left(x\frac{\partial}{\partial x}-\frac{\partial}{\partial x}\right)^{k}$
we will get $2^{k}$ terms, among which those that contain $l$ factors
of $-\frac{\partial}{\partial x}$ would reduce the power of $x^{m}$
by $l$ (note that the factor $x\frac{\partial}{\partial x}$ preserves
the power of $x$). The only term surviving in Eq. (\ref{eq:lemma_important_eq})
is $x^{0}$. Clearly when $k<m$, all the terms have power at least
$m-k$. When $k=m$, the only term surviving is 
\[
\left(-\frac{\partial}{\partial x}\right)^{m}x^{m}=(-1)^{m}m!.
\]
For $k=m+1$, there are $m+1$ surviving term each of which has $m$
factors of $-\frac{\partial}{\partial x}$ and one factor of $x\frac{\partial}{\partial x}$.
They differ by the position where $x\frac{\partial}{\partial x}$
appear. Consider the term with the $i$-th factor being $x\frac{\partial}{\partial x}$;
it is 
\begin{eqnarray*}
\left(-\frac{\partial}{\partial x}\right)^{i-1}x\frac{\partial}{\partial x}\left(-\frac{\partial}{\partial x}\right)^{m+1-i}x^{m} & = & (-1)^{m}\left(\frac{\partial}{\partial x}\right)^{i-1}x\frac{\partial}{\partial x}\frac{m!}{(i-1)!}x^{i-1}\\
 & = & (-1)^{m}(i-1)!\frac{m!}{(i-2)!}\\
 & = & (-1)^{m}m!(i-1).
\end{eqnarray*}
Summing all these terms we get 
\[
\sum_{i=1}^{m+1}(-1)^{m}m!(i-1)=(-1)^{m}m!\left(\begin{array}{c}
m+1\\
2
\end{array}\right).
\]

\end{proof}

\subsection*{Proof of Theorem \ref{thm:asymptotic-Cov}}
\begin{proof}
Let $\beta=\left|\beta\right|e^{i\phi}$, $x\equiv\left|\beta\right|$,
$M=m_{1}+m_{2}+m_{3}+m_{4}$, then we have 

\begin{eqnarray*}
C_{m_{1},m_{2};m_{3},m_{4}}(\beta) & = & \sum_{n}A_{n;m_{1}m_{2}}^{*}A_{n;m_{3}m_{4}}\\
 & = & e^{i\phi(m_{2}+m_{3}-m_{1}-m_{4})}(-1)^{M}\sqrt{m_{1}!m_{2}!m_{3}!m_{4}!}x^{-M}e^{-2x^{2}}\\
 &  & \times\sum_{n}\frac{x^{4n}}{(n!)^{2}}\mathcal{L}_{m_{1}}^{n-m_{1}}(x^{2})\mathcal{L}_{m_{2}}^{n-m_{2}}(x^{2})\mathcal{L}_{m_{3}}^{n-m_{3}}(x^{2})\mathcal{L}_{m_{4}}^{n-m_{4}}(x^{2})
\end{eqnarray*}

Using the explicit formula for the associated Laguerre polynomial
\begin{eqnarray*}
\mathcal{L}_{m}^{n-m}(x^{2}) & = & \sum_{i=0}^{m}\frac{1}{i!}\left(\begin{array}{c}
n\\
m-i
\end{array}\right)(-x^{2})^{i}=\sum_{j=0}^{m}\left(\begin{array}{c}
n\\
j
\end{array}\right)\frac{(-1)^{m-j}}{(m-j)!}x^{2(m-j)},
\end{eqnarray*}
 we find that 
\begin{eqnarray*}
 &  & \sum_{n}\frac{x^{4n}}{(n!)^{2}}\mathcal{L}_{m_{1}}^{n-m_{1}}(x^{2})\mathcal{L}_{m_{2}}^{n-m_{2}}(x^{2})\mathcal{L}_{m_{3}}^{n-m_{3}}(x^{2})\mathcal{L}_{m_{4}}^{n-m_{4}}(x^{2})\\
 & = & \sum_{j_{1}=0}^{m_{1}}\sum_{j_{2}=0}^{m_{2}}\sum_{j_{3}=0}^{m_{3}}\sum_{j_{4}=0}^{m_{4}}\frac{(-1)^{M-j_{1}-j_{2}-j_{3}-j_{4}}x^{2(M-j_{1}-j_{2}-j_{3}-j_{4})}}{(m_{1}-j_{1})!j_{1}!(m_{2}-j_{2})!j_{2}!(m_{3}-j_{3})!j_{3}!(m_{4}-j_{4})!j_{4}!}\\
 &  & \times j_{1}!j_{2}!j_{3}!j_{4}!\sum_{n}\frac{x^{4n}}{(n!)^{2}}\left(\begin{array}{c}
n\\
j_{1}
\end{array}\right)\left(\begin{array}{c}
n\\
j_{2}
\end{array}\right)\left(\begin{array}{c}
n\\
j_{3}
\end{array}\right)\left(\begin{array}{c}
n\\
j_{4}
\end{array}\right)
\end{eqnarray*}
Letting $z=x^{4}$, using Lemma \ref{lem:bino-prod-sum}, we have 

\begin{eqnarray*}
 &  & j_{1}!j_{2}!j_{3}!j_{4}!\sum_{n}\frac{x^{4n}}{(n!)^{2}}\left(\begin{array}{c}
n\\
j_{1}
\end{array}\right)\left(\begin{array}{c}
n\\
j_{2}
\end{array}\right)\left(\begin{array}{c}
n\\
j_{3}
\end{array}\right)\left(\begin{array}{c}
n\\
j_{4}
\end{array}\right)\\
 & = & z^{j_{4}}\frac{\partial^{j_{4}}}{\partial z^{j_{4}}}z^{j_{3}}\frac{\partial^{j_{3}}}{\partial z^{j_{3}}}z^{j_{2}}\frac{\partial^{j_{2}}}{\partial z^{j_{2}}}z^{j_{1}}\frac{\partial^{j_{1}}}{\partial z^{j_{1}}}I_{0}(2\sqrt{z})\\
 & = & \sum_{k_{4}=0}^{j_{4}}\sum_{k_{3}=0}^{j_{3}}\frac{j_{4}!}{k_{4}!(j_{4}-k_{4})!}\frac{j_{3}!}{k_{3}!(j_{3}-k_{3})!}\frac{j_{2}!}{(j_{2}-k_{3})!}\frac{(j_{2}+j_{3}-k_{3})!}{(j_{2}+j_{3}-k_{3}-k_{4})!}(\sqrt{z})^{j_{1}+j_{2}+j_{3}+j_{4}-k_{3}-k_{4}}I_{j_{1}-j_{2}-j_{3}-j_{4}+k_{3}+k_{4}}(2\sqrt{z}).
\end{eqnarray*}
Therefore after some simplification
\begin{eqnarray*}
 &  & \sum_{n}\frac{x^{4n}}{(n!)^{2}}\mathcal{L}_{m_{1}}^{n-m_{1}}(x^{2})\mathcal{L}_{m_{2}}^{n-m_{2}}(x^{2})\mathcal{L}_{m_{3}}^{n-m_{3}}(x^{2})\mathcal{L}_{m_{4}}^{n-m_{4}}(x^{2})\\
 & = & (-1)^{M}x^{2M}\sum_{j_{1}=0}^{m_{1}}\sum_{j_{2}=0}^{m_{2}}\sum_{j_{3}=0}^{m_{3}}\sum_{j_{4}=0}^{m_{4}}\frac{(-1)^{j_{1}+j_{2}+j_{3}+j_{4}}}{(m_{1}-j_{1})!j_{1}!(m_{2}-j_{2})!j_{2}!(m_{3}-j_{3})!j_{3}!(m_{4}-j_{4})!j_{4}!}\\
 &  & \times\sum_{k_{4}=0}^{j_{4}}\sum_{k_{3}=0}^{j_{3}}\frac{j_{4}!}{k_{4}!(j_{4}-k_{4})!}\frac{j_{3}!}{k_{3}!(j_{3}-k_{3})!}\frac{j_{2}!}{(j_{2}-k_{3})!}\frac{(j_{2}+j_{3}-k_{3})!}{(j_{2}+j_{3}-k_{3}-k_{4})!}(x^{2})^{-k_{3}-k_{4}}I_{j_{1}-j_{2}-j_{3}-j_{4}+k_{3}+k_{4}}(2x^{2}).
\end{eqnarray*}
Part of the above formula can be further simplified,
\begin{eqnarray*}
 &  & \sum_{j_{3}=0}^{m_{3}}\sum_{j_{4}=0}^{m_{4}}\frac{(-1)^{j_{3}+j_{4}}}{(m_{3}-j_{3})!j_{3}!(m_{4}-j_{4})!j_{4}!}\\
 &  & \times\sum_{k_{4}=0}^{j_{4}}\sum_{k_{3}=0}^{j_{3}}\frac{j_{4}!}{k_{4}!(j_{4}-k_{4})!}\frac{j_{3}!}{k_{3}!(j_{3}-k_{3})!}\frac{j_{2}!}{(j_{2}-k_{3})!}\frac{(j_{2}+j_{3}-k_{3})!}{(j_{2}+j_{3}-k_{3}-k_{4})!}(x^{2})^{-k_{3}-k_{4}}I_{j_{1}-j_{2}-j_{3}-j_{4}+k_{3}+k_{4}}(2x^{2})\\
 & = & \sum_{k_{4}=0}^{m_{4}}\sum_{k_{3}=0}^{m_{3}}\sum_{j_{3}=k_{3}}^{m_{3}}\sum_{j_{4}=k_{4}}^{m_{4}}\frac{(-1)^{j_{3}+j_{4}}}{(m_{3}-j_{3})!(m_{4}-j_{4})!}\\
 &  & \times\frac{1}{k_{4}!(j_{4}-k_{4})!}\frac{1}{k_{3}!(j_{3}-k_{3})!}\frac{j_{2}!}{(j_{2}-k_{3})!}\frac{(j_{2}+j_{3}-k_{3})!}{(j_{2}+j_{3}-k_{3}-k_{4})!}(x^{2})^{-k_{3}-k_{4}}I_{j_{1}-j_{2}-j_{3}-j_{4}+k_{3}+k_{4}}(2x^{2})\\
 & = & \sum_{k_{4}=0}^{m_{4}}\sum_{k_{3}=0}^{m_{3}}\sum_{j_{3}=0}^{m_{3}-k_{3}}\sum_{j_{4}=0}^{m_{4}-k_{4}}\frac{(-1)^{j_{3}+k_{3}+j_{4}+k_{4}}}{(m_{3}-j_{3}-k_{3})!(m_{4}-j_{4}-k_{4})!}\\
 &  & \times\frac{1}{k_{4}!j_{4}!}\frac{1}{k_{3}!j_{3}!}\frac{j_{2}!}{(j_{2}-k_{3})!}\frac{(j_{2}+j_{3})!}{(j_{2}+j_{3}-k_{4})!}(x^{2})^{-k_{3}-k_{4}}I_{j_{1}-j_{2}-j_{3}-j_{4}}(2x^{2})\\
 & = & \sum_{k_{4}=0}^{m_{4}}\sum_{k_{3}=0}^{m_{3}}\frac{(x^{2})^{-k_{3}-k_{4}}}{k_{3}!k_{4}!}(-1)^{k_{3}+k_{4}}\frac{j_{2}!}{(j_{2}-k_{3})!(m_{3}-k_{3})!(m_{4}-k_{4})!}\\
 &  & \times\sum_{j_{3}=0}^{m_{3}-k_{3}}(-1)^{j_{3}}\frac{(j_{2}+j_{3})!}{(j_{2}+j_{3}-k_{4})!}\left(\begin{array}{c}
m_{3}-k_{3}\\
j_{3}
\end{array}\right)\sum_{j_{4}=0}^{m_{4}-k_{4}}(-1)^{j_{4}}\left(\begin{array}{c}
m_{4}-k_{4}\\
j_{4}
\end{array}\right)I_{j_{1}-j_{2}-j_{3}-j_{4}}(2x^{2}).
\end{eqnarray*}
Now 
\begin{eqnarray*}
 &  & \sum_{n}\frac{x^{4n}}{(n!)^{2}}\mathcal{L}_{m_{1}}^{n-m_{1}}(x^{2})\mathcal{L}_{m_{2}}^{n-m_{2}}(x^{2})\mathcal{L}_{m_{3}}^{n-m_{3}}(x^{2})\mathcal{L}_{m_{4}}^{n-m_{4}}(x^{2})\\
 & = & (-1)^{M}x^{2M}\frac{1}{m_{1}!m_{2}!}\sum_{k_{4}=0}^{m_{4}}\sum_{k_{3}=0}^{m_{3}}\frac{(x^{2})^{-k_{3}-k_{4}}}{k_{3}!k_{4}!(m_{3}-k_{3})!(m_{4}-k_{4})!}(-1)^{k_{3}+k_{4}}\\
 &  & \times\sum_{j_{1}=0}^{m_{1}}(-1)^{j_{1}}\left(\begin{array}{c}
m_{1}\\
j_{1}
\end{array}\right)\sum_{j_{2}=0}^{m_{2}}(-1)^{j_{2}}\left(\begin{array}{c}
m_{2}\\
j_{2}
\end{array}\right)\frac{j_{2}!}{(j_{2}-k_{3})!}\\
 &  & \times\sum_{j_{3}=0}^{m_{3}-k_{3}}(-1)^{j_{3}}\frac{(j_{2}+j_{3})!}{(j_{2}+j_{3}-k_{4})!}\left(\begin{array}{c}
m_{3}-k_{3}\\
j_{3}
\end{array}\right)\sum_{j_{4}=0}^{m_{4}-k_{4}}(-1)^{j_{4}}\left(\begin{array}{c}
m_{4}-k_{4}\\
j_{4}
\end{array}\right)I_{j_{1}-j_{2}-j_{3}-j_{4}}(2x^{2}).
\end{eqnarray*}

We now focus on one term in the double summation $\sum_{k_{4}=0}^{m_{4}}\sum_{k_{3}=0}^{m_{3}}$,
i.e., the summand with fixed $k_{3}$ and $k_{4}$. It is known that
for large $z$, 
\[
I_{\nu}(z)\sim\frac{e^{z}}{\sqrt{2\pi z}}\left[1-\frac{4\nu^{2}-1}{8z}+\frac{(4\nu^{2}-1)(4\nu^{2}-9)}{2!(8z)^{2}}+\cdots+(-1)^{l}\frac{\prod_{i=1}^{l}\left[4\nu^{2}-(2i-1)^{2}\right]}{l!(8z)^{l}}+\cdots\right],
\]
in our case 
\[
I_{j_{1}-j_{2}-j_{3}-j_{4}}(2x^{2})\sim\frac{e^{2x^{2}}}{2x\sqrt{\pi}}\left[1-\frac{4(j_{1}-j_{2}-j_{3}-j_{4})^{2}-1}{16x^{2}}\cdots+(-1)^{l}\frac{\prod_{i=1}^{l}\left[4(j_{1}-j_{2}-j_{3}-j_{4})^{2}-(2i-1)^{2}\right]}{l!(4x)^{2l}}+\cdots\right].
\]
The expansion of $I_{j_{1}-j_{2}-j_{3}-j_{4}}(2x^{2})$ contains polynomials
of the form $j_{1}^{p_{1}}j_{2}^{p_{2}}j_{3}^{p_{3}}j_{4}^{p_{4}}$.
Note also $\frac{j_{2}!}{(j_{2}-k_{3})!}$ is a polynomial of $j_{2}$
of degree $k_{3}$ and $\frac{(j_{2}+j_{3})!}{(j_{2}+j_{3}-k_{4})!}$
is polynomial of $(j_{2}+j_{3})$ of degree $k_{4}$. So overall the
summand of the quadruple summation $\sum_{j_{1}=0}^{m_{1}}\sum_{j_{2}=0}^{m_{2}}\sum_{j_{3}=0}^{m_{3}-k_{3}}\sum_{j_{4}=0}^{m_{4}-k_{4}}$
is a combination of polynomials of the form $j_{1}^{p_{1}}j_{2}^{p_{2}}j_{3}^{p_{3}}j_{4}^{p_{4}}$.
Due to Lemma \ref{lem:moments-binom}, the terms $j_{1}^{p_{1}}j_{2}^{p_{2}}j_{3}^{p_{3}}j_{4}^{p_{4}}$
that gives non-zero contribution are those with $p_{1}\ge m_{1}$,
$p_{2}\ge m_{2}$, $p_{3}\ge m_{3}-k_{3}$, and $p_{4}\ge m_{4}-k_{4}$.
We try to find such terms with the lowest power in $\frac{1}{x}$,
i.e., to find the smallest $l$ such that the  expression 
\[
\frac{j_{2}!}{(j_{2}-k_{3})!}\frac{(j_{2}+j_{3})!}{(j_{2}+j_{3}-k_{4})!}\prod_{i=1}^{l}\left[4(j_{1}-j_{2}-j_{3}-j_{4})^{2}-(2i-1)^{2}\right]
\]
 contains a term like $j_{1}^{m_{1}}j_{2}^{m_{2}}j_{3}^{m_{3}-k_{3}}j_{4}^{m_{4}-k_{4}}$
or of even higher order. Since

\[
\frac{j_{2}!}{(j_{2}-k_{3})!}\frac{(j_{2}+j_{3})!}{(j_{2}+j_{3}-k_{4})!}\prod_{i=1}^{l}\left[4(j_{1}-j_{2}-j_{3}-j_{4})^{2}-(2i-1)^{2}\right]=j_{2}^{k_{3}}(j_{2}+j_{3})^{k_{4}}4^{l}(j_{1}-j_{2}-j_{3}-j_{4})^{2l}+(\text{lower order terms}),
\]
we must require 
\[
k_{3}+k_{4}+2l\ge m_{1}+m_{2}+m_{3}-k_{3}+m_{4}-k_{4},
\]
i.e., 
\[
2(l+k_{3}+k_{4})\ge m_{1}+m_{2}+m_{3}+m_{4}=M.
\]
 Thus the smallest $l$ should be 
\[
l_{*}=\begin{cases}
\frac{M}{2}-k_{3}-k_{4}, & \text{if }M\text{ even};\\
\frac{M+1}{2}-k_{3}-k_{4}, & \text{if }M\text{ odd}.
\end{cases}
\]
So if we neglect terms that either give zero contribution to the quadruple
sum over $j_{i}$ or are not of the leading order in $\frac{1}{x}$,
\begin{eqnarray*}
 &  & \frac{j_{2}!}{(j_{2}-k_{3})!}\frac{(j_{2}+j_{3})!}{(j_{2}+j_{3}-k_{4})!}I_{j_{1}-j_{2}-j_{3}-j_{4}}(2x^{2})\\
 & \sim & j_{2}^{k_{3}}(j_{2}+j_{3})^{k_{4}}(-1)^{l_{*}}\frac{4^{l_{*}}(j_{1}-j_{2}-j_{3}-j_{4})^{2l_{*}}}{l_{*}!(4x)^{2l_{*}}}\frac{e^{2x^{2}}}{2x\sqrt{\pi}}\\
 & = & j_{2}^{k_{3}}(j_{2}+j_{3})^{k_{4}}(j_{1}-j_{2}-j_{3}-j_{4})^{2l_{*}}\frac{(-1)^{l_{*}}}{l_{*}!4^{l_{*}}x^{2l_{*}}}\frac{e^{2x^{2}}}{2x\sqrt{\pi}}.
\end{eqnarray*}
\textbf{ When $M$ is even,} $2(l_{*}+k_{3}+k_{4})=M$, so 

\begin{eqnarray*}
 &  & \frac{j_{2}!}{(j_{2}-k_{3})!}\frac{(j_{2}+j_{3})!}{(j_{2}+j_{3}-k_{4})!}I_{j_{1}-j_{2}-j_{3}-j_{4}}(2x^{2})\\
 & \sim & j_{2}^{k_{3}}\sum_{\mu=0}^{k_{4}}\left(\begin{array}{c}
k_{4}\\
\mu
\end{array}\right)j_{2}^{\mu}j_{3}^{(k_{4}-\mu)}\\
 &  & \times(-1)^{m_{2}+m_{3}+m_{4}-2(k_{3}+k_{4})}j_{1}^{m_{1}}j_{2}^{m_{2}-k_{3}-\mu}j_{3}^{m_{3}-k_{3}-k_{4}+\mu}j_{4}^{m_{4}-k_{4}}\\
 &  & \times\left(\begin{array}{c}
M-2k_{3}-2k_{4}\\
m_{1},m_{2}-k_{3}-\mu,m_{3}-k_{3}-k_{4}+\mu,m_{4}-k_{4}
\end{array}\right)\frac{(-1)^{l_{*}}}{l_{*}!4^{l_{*}}x^{2l_{*}}}\frac{e^{2x^{2}}}{2x\sqrt{\pi}}\\
 & = & (-1)^{M-m_{1}}\sum_{\mu=0}^{k_{4}}\left(\begin{array}{c}
k_{4}\\
\mu
\end{array}\right)\left(\begin{array}{c}
M-2k_{3}-2k_{4}\\
m_{1},\, m_{2}-k_{3}-\mu,\, m_{3}-k_{3}-k_{4}+\mu,\, m_{4}-k_{4}
\end{array}\right)j_{1}^{m_{1}}j_{2}^{m_{2}}j_{3}^{m_{3}-k_{3}}j_{4}^{m_{4}-k_{4}}\\
 &  & \times\frac{(-1)^{M/2-k_{3}-k_{4}}}{(M/2-k_{3}-k_{4})!2^{(M-2k_{3}-2k_{4})}x^{(M-2k_{3}-2k_{4})}}\frac{e^{2x^{2}}}{2x\sqrt{\pi}},
\end{eqnarray*}
 where $\left(\begin{array}{c}
n\\
k_{1},\, k_{2},\,\cdots,\, k_{m}
\end{array}\right)\equiv\frac{n!}{k_{1}!k_{2}!\cdots k_{m}!}$. 

Using Lemma \ref{lem:moments-binom}, 
\begin{eqnarray*}
 &  & \sum_{j_{1}=0}^{m_{1}}(-1)^{j_{1}}\left(\begin{array}{c}
m_{1}\\
j_{1}
\end{array}\right)\sum_{j_{2}=0}^{m_{2}}(-1)^{j_{2}}\left(\begin{array}{c}
m_{2}\\
j_{2}
\end{array}\right)\sum_{j_{3}=0}^{m_{3}-k_{3}}(-1)^{j_{3}}\left(\begin{array}{c}
m_{3}-k_{3}\\
j_{3}
\end{array}\right)\sum_{j_{4}=0}^{m_{4}-k_{4}}(-1)^{j_{4}}\left(\begin{array}{c}
m_{4}-k_{4}\\
j_{4}
\end{array}\right)j_{1}^{m_{1}}j_{2}^{m_{2}}j_{3}^{m_{3}-k_{3}}j_{4}^{m_{4}-k_{4}}\\
 & = & \sum_{j_{1}=0}^{m_{1}}(-1)^{j_{1}}\left(\begin{array}{c}
m_{1}\\
j_{1}
\end{array}\right)j_{1}^{m_{1}}\sum_{j_{2}=0}^{m_{2}}(-1)^{j_{2}}\left(\begin{array}{c}
m_{2}\\
j_{2}
\end{array}\right)j_{2}^{m_{2}}\sum_{j_{3}=0}^{m_{3}-k_{3}}(-1)^{j_{3}}\left(\begin{array}{c}
m_{3}-k_{3}\\
j_{3}
\end{array}\right)j_{3}^{m_{3}-k_{3}}\sum_{j_{4}=0}^{m_{4}-k_{4}}(-1)^{j_{4}}\left(\begin{array}{c}
m_{4}-k_{4}\\
j_{4}
\end{array}\right)j_{4}^{m_{4}-k_{4}}\\
 & = & (-1)^{M-k_{3}-k_{4}}m_{1}!m_{2}!(m_{3}-k_{3})!(m_{4}-k_{4})!.
\end{eqnarray*}
Plugging back to the expression of $\sum_{n}\frac{x^{4n}}{(n!)^{2}}\mathcal{L}_{m_{1}}^{n-m_{1}}(x^{2})\mathcal{L}_{m_{2}}^{n-m_{2}}(x^{2})\mathcal{L}_{m_{3}}^{n-m_{3}}(x^{2})\mathcal{L}_{m_{4}}^{n-m_{4}}(x^{2})$
we eventually get 
\begin{eqnarray*}
 &  & \sum_{n}\frac{x^{4n}}{(n!)^{2}}\mathcal{L}_{m_{1}}^{n-m_{1}}(x^{2})\mathcal{L}_{m_{2}}^{n-m_{2}}(x^{2})\mathcal{L}_{m_{3}}^{n-m_{3}}(x^{2})\mathcal{L}_{m_{4}}^{n-m_{4}}(x^{2})\\
 & \sim & (-1)^{m_{1}+M/2}e^{2x^{2}}x^{M-1}\frac{1}{2^{M+1}}\frac{1}{\sqrt{\pi}}\sum_{k_{4}=0}^{m_{4}}\sum_{k_{3}=0}^{m_{3}}\frac{(-1)^{k_{3}+k_{4}}2^{2(k_{3}+k_{4})}}{k_{3}!k_{4}!(M/2-k_{3}-k_{4})!}\\
 &  & \times\sum_{\mu=0}^{k_{4}}\left(\begin{array}{c}
k_{4}\\
\mu
\end{array}\right)\left(\begin{array}{c}
M-2k_{3}-2k_{4}\\
m_{1},\, m_{2}-k_{3}-\mu,\, m_{3}-k_{3}-k_{4}+\mu,\, m_{4}-k_{4}
\end{array}\right).
\end{eqnarray*}
Finally, we have the leading order contribution for the even $M$
case:
\begin{eqnarray*}
C_{m_{1},m_{2};m_{3},m_{4}}(\beta) & \sim & x^{-1}e^{i\phi(m_{2}+m_{3}-m_{1}-m_{4})}\sqrt{m_{1}!m_{2}!m_{3}!m_{4}!}(-1)^{m_{1}+M/2}\frac{1}{2^{M+1}\sqrt{\pi}}\\
 &  & \times\sum_{k_{4}=0}^{m_{4}}\sum_{k_{3}=0}^{m_{3}}\frac{(-1)^{k_{3}+k_{4}}2^{2(k_{3}+k_{4})}}{k_{3}!k_{4}!(M/2-k_{3}-k_{4})!}\\
 &  & \times\sum_{\mu=0}^{k_{4}}\left(\begin{array}{c}
k_{4}\\
\mu
\end{array}\right)\left(\begin{array}{c}
M-2k_{3}-2k_{4}\\
m_{1},\, m_{2}-k_{3}-\mu,\, m_{3}-k_{3}-k_{4}+\mu,\, m_{4}-k_{4}
\end{array}\right)\\
 & = & \frac{g(m_{1},\, m_{2},\, m_{3},\, m_{4},\,\phi)}{\left|\beta\right|}.
\end{eqnarray*}

\textbf{When $M$ is odd,} $2(l_{*}+k_{3}+k_{4})=M+1$. In this case
five terms give non-zero contribution under the quadruple sum of $j_{i}$,
which are $P_{1}\equiv j_{1}^{m_{1}+1}j_{2}^{m_{2}}j_{3}^{m_{3}-k_{3}}j_{4}^{m_{4}-k_{4}}$,
$P_{2}\equiv j_{1}^{m_{1}}j_{2}^{m_{2}+1}j_{3}^{m_{3}-k_{3}}j_{4}^{m_{4}-k_{4}}$,
$P_{3}\equiv j_{1}^{m_{1}}j_{2}^{m_{2}}j_{3}^{m_{3}-k_{3}+1}j_{4}^{m_{4}-k_{4}}$,
$P_{4}\equiv j_{1}^{m_{1}}j_{2}^{m_{2}}j_{3}^{m_{3}-k_{3}}j_{4}^{m_{4}-k_{4}+1}$
and $P_{5}\equiv j_{1}^{m_{1}}j_{2}^{m_{2}}j_{3}^{m_{3}-k_{3}}j_{4}^{m_{4}-k_{4}}$.
$P_{1},\,\cdots,\: P_{4}$ are the highest-order terms about the variables
$j_{i}$ in the summand and $P_{5}$ is the next highest order. Let
us write 
\begin{eqnarray*}
\frac{j_{2}!}{(j_{2}-k_{3})!}\frac{(j_{2}+j_{3})!}{(j_{2}+j_{3}-k_{4})!}I_{j_{1}-j_{2}-j_{3}-j_{4}}(2x^{2}) & \sim & \frac{(-1)^{l_{*}}}{l_{*}!4^{l_{*}}x^{2l_{*}}}\frac{e^{2x^{2}}}{2x\sqrt{\pi}}\sum_{\nu=1}^{5}\lambda_{\nu}P_{\nu}.
\end{eqnarray*}
The coefficients $\lambda_{\nu}$ are essentially combinatoric factors
and it is not difficult to work them out, although the process can
be long and tedious. Eventually we find, 
\[
\lambda_{1}=(-1)^{m_{1}+1}\sum_{\mu=0}^{k_{4}}\left(\begin{array}{c}
k_{4}\\
\mu
\end{array}\right)\left(\begin{array}{c}
M+1-2k_{3}-2k_{4}\\
m_{1}+1,\, m_{2}-k_{3}-\mu,\, m_{3}-k_{3}-k_{4}+\mu,\, m_{4}-k_{4}
\end{array}\right),
\]

\[
\lambda_{2}=(-1)^{m_{1}}\sum_{\mu=0}^{k_{4}}\left(\begin{array}{c}
k_{4}\\
\mu
\end{array}\right)\left(\begin{array}{c}
M+1-2k_{3}-2k_{4}\\
m_{1},\, m_{2}-k_{3}-\mu+1,\, m_{3}-k_{3}-k_{4}+\mu,\, m_{4}-k_{4}
\end{array}\right),
\]

\[
\lambda_{3}=(-1)^{m_{1}}\sum_{\mu=0}^{k_{4}}\left(\begin{array}{c}
k_{4}\\
\mu
\end{array}\right)\left(\begin{array}{c}
M+1-2k_{3}-2k_{4}\\
m_{1},\, m_{2}-k_{3}-\mu,\, m_{3}-k_{3}-k_{4}+\mu+1,\, m_{4}-k_{4}
\end{array}\right),
\]
\[
\lambda_{4}=(-1)^{m_{1}}\sum_{\mu=0}^{k_{4}}\left(\begin{array}{c}
k_{4}\\
\mu
\end{array}\right)\left(\begin{array}{c}
M+1-2k_{3}-2k_{4}\\
m_{1},\, m_{2}-k_{3}-\mu,\, m_{3}-k_{3}-k_{4}+\mu,\, m_{4}-k_{4}+1
\end{array}\right),
\]
\begin{eqnarray*}
\lambda_{5} & = & \sum_{\mu=0}^{k_{4}}\frac{-k_{3}(k_{3}-1)}{2}\left(\begin{array}{c}
k_{4}\\
\mu
\end{array}\right)\left(\begin{array}{c}
M+1-2k_{3}-2k_{4}\\
m_{1},\, m_{2}-k_{3}-\mu+1,\, m_{3}-k_{3}-k_{4}+\mu,\, m_{4}-k_{4}
\end{array}\right)\\
 &  & +\sum_{\mu=0}^{k_{4}-1}\frac{-k_{4}(k_{4}-1)}{2}\left(\begin{array}{c}
k_{4}-1\\
\mu
\end{array}\right)\left(\begin{array}{c}
M+1-2k_{3}-2k_{4}\\
m_{1},\, m_{2}-k_{3}-\mu,\, m_{3}-k_{3}-k_{4}+\mu+1,\, m_{4}-k_{4}
\end{array}\right).
\end{eqnarray*}
The \textbf{key point} to notice is that because $2(l_{*}+k_{3}+k_{4})=M+1$,
now the leading term in $\frac{1}{x}$ is 
\begin{eqnarray*}
\frac{(-1)^{l_{*}}}{l_{*}!4^{l_{*}}x^{2l_{*}}}\frac{e^{2x^{2}}}{2x\sqrt{\pi}} & = & \frac{(-1)^{(M+1)/2-k_{3}-k_{4}}}{\left((M+1)/2-k_{3}-k_{4}\right)!2^{(M+1-2k_{3}-2k_{4})}x^{(M+1-2k_{3}-2k_{4})}}\frac{e^{2x^{2}}}{2x\sqrt{\pi}}\sim\frac{1}{x^{(M+1-2k_{3}-2k_{4})}}\frac{e^{2x^{2}}}{2x\sqrt{\pi}},
\end{eqnarray*}
which is one order higher in $\frac{1}{x}$ compared to the even $M$
case. 

Using Lemma \ref{lem:moments-binom}, 
\begin{eqnarray*}
 &  & \sum_{j_{1}=0}^{m_{1}}(-1)^{j_{1}}\left(\begin{array}{c}
m_{1}\\
j_{1}
\end{array}\right)\sum_{j_{2}=0}^{m_{2}}(-1)^{j_{2}}\left(\begin{array}{c}
m_{2}\\
j_{2}
\end{array}\right)\sum_{j_{3}=0}^{m_{3}-k_{3}}(-1)^{j_{3}}\left(\begin{array}{c}
m_{3}-k_{3}\\
j_{3}
\end{array}\right)\sum_{j_{4}=0}^{m_{4}-k_{4}}(-1)^{j_{4}}\left(\begin{array}{c}
m_{4}-k_{4}\\
j_{4}
\end{array}\right)\sum_{\nu=1}^{5}\lambda_{\nu}P_{\nu}\\
 & = & (-1)^{M-k_{3}-k_{4}}m_{1}!m_{2}!(m_{3}-k_{3})!(m_{4}-k_{4})!\\
 &  & \times\left[\lambda_{5}+\left(\begin{array}{c}
m_{1}+1\\
2
\end{array}\right)\lambda_{1}+\left(\begin{array}{c}
m_{2}+1\\
2
\end{array}\right)\lambda_{2}+\left(\begin{array}{c}
m_{3}-k_{3}+1\\
2
\end{array}\right)\lambda_{3}+\left(\begin{array}{c}
m_{4}-k_{4}+1\\
2
\end{array}\right)\lambda_{4}\right].
\end{eqnarray*}

Now 
\begin{eqnarray*}
 &  & \sum_{n}\frac{x^{4n}}{(n!)^{2}}\mathcal{L}_{m_{1}}^{n-m_{1}}(x^{2})\mathcal{L}_{m_{2}}^{n-m_{2}}(x^{2})\mathcal{L}_{m_{3}}^{n-m_{3}}(x^{2})\mathcal{L}_{m_{4}}^{n-m_{4}}(x^{2})\\
 & \sim & (-1)^{(M+1)/2}e^{2x^{2}}x^{M-2}\frac{1}{2^{M+2}\sqrt{\pi}}\sum_{k_{4}=0}^{m_{4}}\sum_{k_{3}=0}^{m_{3}}\frac{(-1)^{k_{3}+k_{4}}2^{2(k_{3}+k_{4})}}{k_{3}!k_{4}!\left((M+1)/2-k_{3}-k_{4}\right)!}\\
 &  & \times\left[\lambda_{5}+\left(\begin{array}{c}
m_{1}+1\\
2
\end{array}\right)\lambda_{1}+\left(\begin{array}{c}
m_{2}+1\\
2
\end{array}\right)\lambda_{2}+\left(\begin{array}{c}
m_{3}-k_{3}+1\\
2
\end{array}\right)\lambda_{3}+\left(\begin{array}{c}
m_{4}-k_{4}+1\\
2
\end{array}\right)\lambda_{4}\right].
\end{eqnarray*}
Finally,
\begin{eqnarray*}
C_{m_{1},m_{2};m_{3},m_{4}}(\beta) & \sim & -x^{-2}e^{i\phi(m_{2}+m_{3}-m_{1}-m_{4})}\sqrt{m_{1}!m_{2}!m_{3}!m_{4}!}(-1)^{(M+1)/2}\frac{1}{2^{M+2}\sqrt{\pi}}\\
 &  & \times\sum_{k_{4}=0}^{m_{4}}\sum_{k_{3}=0}^{m_{3}}\frac{(-1)^{k_{3}+k_{4}}2^{2(k_{3}+k_{4})}}{k_{3}!k_{4}!\left((M+1)/2-k_{3}-k_{4}\right)!}\\
 &  & \times\left[\lambda_{5}+\left(\begin{array}{c}
m_{1}+1\\
2
\end{array}\right)\lambda_{1}+\left(\begin{array}{c}
m_{2}+1\\
2
\end{array}\right)\lambda_{2}+\left(\begin{array}{c}
m_{3}-k_{3}+1\\
2
\end{array}\right)\lambda_{3}+\left(\begin{array}{c}
m_{4}-k_{4}+1\\
2
\end{array}\right)\lambda_{4}\right]\\
 & = & \frac{g(m_{1},\, m_{2},\, m_{3},\, m_{4},\,\phi)}{\left|\beta\right|^{2}}.
\end{eqnarray*}

In summary, we have thus proved that for large $\left|\beta\right|$,
\[
C_{m_{1}m_{2},m_{3}m_{4}}(\beta)\sim\begin{cases}
g(m_{1},\, m_{2},\, m_{3},\, m_{4},\,\phi)/\left|\beta\right|, & \sum_{i=1}^{4}m_{i}\ \text{is even};\\
g(m_{1},\, m_{2},\, m_{3},\, m_{4},\,\phi)/\left|\beta\right|^{2}, & \sum_{i=1}^{4}m_{i}\ \text{is odd};
\end{cases}
\]

In fact our technique can be used to prove the general asymptotic
result

\[
\sum_{n=0}^{\infty}\frac{1}{(n!)^{2}}x^{4n}\prod_{i}\mathcal{L}_{m_{i}}^{n-m_{i}}(x^{2})\sim\begin{cases}
x^{\sum_{i}m_{i}-1}, & \ \sum_{i}m_{i}\ \text{is even};\\
x^{\sum_{i}m_{i}-2}, & \ \sum_{i}m_{i}\ \text{is odd}.
\end{cases}
\]

\end{proof}

\end{widetext}

\bibliographystyle{apsrev4-1}
\bibliography{reference_list}

\end{document}